%% file: main_single_output.tex
\documentclass[a4paper]{article}

\usepackage[colorlinks=true,citecolor=magenta,linkcolor=blue]{hyperref}
\usepackage{booktabs} 
\usepackage[a4paper]{geometry}
\usepackage[ruled]{algorithm2e} 
\usepackage{caption}
\usepackage{amsfonts}
\usepackage{amsmath}
\usepackage{amssymb}
\usepackage{amsthm}
\usepackage{graphicx}
\usepackage{mathtools}
\usepackage{tikz}
\usetikzlibrary{matrix, arrows,fit,shapes.gates.logic.US,shapes.gates.logic.IEC,calc,backgrounds,decorations.pathmorphing}
\usepackage{pgfplots}
\usepackage{authblk}

\renewcommand{\emph}[1]{{\textbf{#1}}}

\newcommand{\AND}{\textsc{And}}
\newcommand{\OR}{\textsc{Or}}
\newcommand{\AOP}{\textsc{And}-\textsc{Or} Path}
\newcommand{\aop}{\textsc{And}-\textsc{Or} path}

\newcommand{\blank}[1]{-}

\newcommand{\customalign}[2]{\mathmakebox[\widthof{#1}]{#2}}

\DeclareMathOperator{\N}{\mathbb N}

\DeclareMathOperator{\Ceta}{\zeta}

\theoremstyle{plain}
\newtheorem{theorem}{Theorem}
\newtheorem{lemma}[theorem]{Lemma}

\theoremstyle{definition}
\newtheorem{definition}[theorem]{Definition}
\newenvironment{claim}{\par\noindent\textit{Claim:}}{}
\newenvironment{proof_of_claim}{\par\noindent\textit{Proof of the claim:}}{}

\begin{document}

\newtheorem{rem}[theorem]{Remark}
\newtheorem{problem}[theorem]{Problem}

\title{Faster Carry Bit Computation for Adder Circuits with Prescribed Arrival Times}

\author{Ulrich Brenner, Anna Hermann}
\affil{Research Institute for Discrete Mathematics, University of Bonn}
\date{\vspace{-5ex}}

\maketitle

\begin{abstract}
We consider the fundamental problem of constructing fast circuits for the carry bit computation in binary addition.
Up to a small additive constant, the carry bit computation reduces to computing
an \textsc{And}-\textsc{Or} path, i.e., a formula of type
$t_0 \land (t_1 \lor (t_2 \land ( \dots t_{m-1}) \dots )$ or
$t_0 \lor (t_1 \land (t_2 \lor ( \dots t_{m-1}) \dots )$.
We present an algorithm that computes the fastest known Boolean circuit for an \textsc{And}-\textsc{Or} path
with given arrival times $a(t_0), \dotsc, a(t_{m-1})$ for the input signals.
Our objective function is delay, a natural generalization of depth with respect to arrival times.
The maximum delay of the circuit we compute is
$\log_2 W + \log_2 \log_2 m + \log_2 \log_2 \log_2 m + 4.3$,
where $W := \sum_{i = 0}^{m-1} 2^{a(t_i)}$.
Note that $\lceil \log_2 W \rceil$ is a lower bound on the delay of any circuit
depending on inputs $t_0, \dotsc, t_{m-1}$
with prescribed arrival times.
Our method yields the fastest circuits for \textsc{And}-\textsc{Or} paths, carry bit computation and adders in terms of delay known so far.
\end{abstract}

%
%



\input{intro.tex}

\input{bounding_V.tex}

\input{bounding_delay.tex}
\input{general_functions.tex}

\pagebreak

\bibliographystyle{plainurl}
\bibliography{extended}

\end{document}

%% file: intro.tex
\section{Introduction}

An {\aop} is a Boolean formula of type
\[
   t_0 \land (t_1 \lor (t_2 \land ( \dots t_{m-1}) \dots )
   \quad
   \text{ or }
   \quad
   t_0 \lor (t_1 \land (t_2 \lor ( \dots t_{m-1}) \dots )\,.
\]
We assume that for each Boolean input variable $t_i$ ($i \in \{0,\dots,m-1\}$),
an arrival time $a(t_i) \in \N$ is given. Our goal is to find a Boolean circuit
using only \AND~and \OR~gates
that computes the Boolean function of a given \aop~and minimizes the maximum delay
of the inputs.
Here, the delay of an input $t_i$ in a Boolean circuit is its arrival time $a(t_i)$
plus the length of a maximum directed path in the circuit starting
at $t_i$. Thus, the concept of delay minimization generalizes the
concept of depth minimization.
When only depth is considered, one has to assume that all input signals are available
at the same time (uniform arrival times), which is not the case on real-world chips.
Hence, taking non-uniform arrival times into account leads to a much more realistic problem
formulation.

\subsection{Applications of \AOP~Optimization}

Realizing \aop s with fast circuits is important for VLSI design in different aspects.

Computing fast \aop s is, up to a small additive constant,
equivalent to constructing fast adder circuits.
To see this, suppose we want to compute the sum of two binary numbers
$x = (x_{r-1}, \dotsc, x_0)$ and $y = (y_{r-1}, \dotsc, y_0)$ with most significant bit $r-1$.
Starting with $c_1 = (x_0 \land y_0)$,
the carry bits can be computed recursively via
\begin{align*}
   c_{i+1} & = (x_i \land y_i) \lor \Big((x_i \lor y_i) \land c_i\Big) \\
                  & =  \underbrace{(x_i \land y_i)}_{t_0} \lor \bigg(\underbrace{(x_i \lor y_i)}_{t_1} \land
                        \Big(\underbrace{(x_{i-1} \land y_{i-1})}_{t_2} \lor \big(\underbrace{(x_{i-1} \lor y_{i-1})}_{t_3} \land \dotsc \big)\Big)\bigg)
\end{align*}
for $0 < i < r$.
The computation of the $i$th carry bit is hence essentially the evaluation of an {\aop } of length $2i$.
Once all carry bits are known, the sum can be computed with an additional delay of $2$.
Adders with non-uniform input arrival times occur, e.g.,\
as a part of multiplication units (see Zimmermann~\cite{ZimmermannDiss}).
Depth optimization of adders (and thus of {\aop}s for the carry bit computation)
is a classical and well-studied optimization problem,
see 
Sklansky~\cite{sklansky1960conditional}, 
Brent~\cite{brent1970addition}, 
Khrapchenko~\cite{krapchenko1970asymptotic}, 
Kogge and Stone~\cite{kogge1973parallel}, 
Ladner and Fischer~\cite{ladner1980parallel}, and 
Brent and Kung~\cite{brent1982regular}. 

Another application of \aop~optimization is the comparison of binary numbers
since lexicographic comparison can be expressed as an \aop~(see, e.g., Grinchuk~\cite{Grinchuk}).

More generally, \aop~optimization is used to speed up timing-critical paths on VLSI chips.
Note that the most critical path $P$ on a chip
can be decomposed into, e.g., \AND~gates and inverter gates.
Using De Morgan's laws, the inverters can be pushed to the inputs of $P$,
making $P$ a path consisting of \AND~and \OR~gates (not necessarily alternating).
In \cite{Werber:BonnLogic}, Werber et al.\
successfully applied their \aop~optimization algorithm presented in \cite{Werber}
for iteratively improving such paths in a late stage of physical design.
In this context,
it is essential that the objective function of the algorithm used is delay instead of depth
since typically, input signals of the most critical path
will not arrive simultaneously.

\subsection{Previous Work}

For \aop s with uniform arrival times (i.e., only depth is considered),
the currently fastest Boolean circuit has been proposed by
Grinchuk~\cite{Grinchuk} reaching a depth of
$\log_2 m + \log_2 \log_2 m + \mathcal O(1)$.
This is close to the best known lower bounds on depth:
Khrapchenko~\cite{Khrapchenko} showed that
any circuit for an \aop~has a depth of at least
$\log_2 m + 0.15 \log_2 \log_2 \log_2 m + \Theta(1)$. The result is based
on a lower bound of
$\Theta\left(\frac{m \log_2 m \log_2 \log_2 m }{\log_2 \log_2 \log_2 \log_2 m }\right)$
on the product of size and depth of a Boolean formula
for an {\aop } (see Commentz-Walter and Sattler~\cite{Commentz-Walter80}).
For monotone circuits, i.e., circuits without negations (and the circuit built in
\cite{Grinchuk} is a monotone circuit), this lower bound can be improved to
$\Theta(m \log_2^2 m)$ (see Commentz-Walter~\cite{Commentz-Walter79}).
This directly implies a
lower bound of $\log_2 m + \log_2 \log_2 m + \Theta(1)$ on the depth
of a monotone circuit for an \aop.


\begin{table}
 \centering
   \begin{tabular}{llll}
    Author & Upper bound on delay & Size & Maximum fanout\\
    \toprule
    \cite{Werber}; \cite{Spirkl} & $1.441 \log_2 W + 2.674$ & $\mathcal O (m)$ & 2\\
    \cite{RautenbachEtal03}; \cite{SpirklMaster} & $(1 + \varepsilon) \left\lceil \log_2 W \right\rceil + \frac{3}{\varepsilon} + 5$ & $\mathcal O(\frac{m}{\varepsilon})$ & $2$\\
    \cite{RautenbachEtal03}; \cite{SpirklMaster} & $(1 + \varepsilon) \left\lceil \log_2 W \right\rceil + \frac{3}{\varepsilon} + 5$ & $\mathcal O(m)$ & $2^\frac{1}{\varepsilon}$\\
    \cite{SpirklMaster} & $\left\lceil \log_2 W \right\rceil + 2 \sqrt{2 \log_2 m - 2} + 6$ &  $\mathcal O(m \sqrt{\log_2 m})$ & $2$\\
    \cite{SpirklMaster} & $\left\lceil \log_2 W \right\rceil + 2 \sqrt{2 \log_2 m - 2} + 6$ &  $\mathcal O(m)$ & $2^{\sqrt{2 \log_2 m - 2}} + 1$\\
    Here & $ \log_2 W + \log_2 \log_2 m $ & $\mathcal O(m \log_2m \log_2 \log_2 m )$ & $\mathcal O(\log_2 m)$\\
         & \phantom{A} $+ \log_2 \log_2 \log_2 m + 4.3$ & & \\
   \bottomrule
   \end{tabular}
   \caption{Known upper bounds on delay of \aop s with non-uniform arrival times.}
   \label{non-uniform bounds}
\end{table}

For \aop s with non-uniform arrival times $a(t_0), \dotsc, a(t_{m-1})$,
the value $\lceil \log_2 W \rceil$
is a lower bound on the achievable delay,
where $W:= \sum_{i=0}^{m-1} 2^{a(t_i)}$ (see, e.g., Rautenbach et al.\ ~\cite{Werber}).
No stronger lower bounds on delay are known.
Rautenbach et al.~\cite{Werber} presented an algorithm
computing a Boolean circuit for an {\aop } with delay at most
$1.441 \log_2 W + 3$.
This delay bound was improved to $1.441 \log_2 W + 2.674$ by Held and Spirkl~\cite{Spirkl}.
In both of these circuits, so-called 2-input prefix gates are used, and
it can be shown that any \aop~realization based on prefix gates has a
delay of at least $\log_\varphi\left( \sum_{i=0}^{m-1} \varphi^{a(t_i)} \right)$
where $\varphi = \frac{1+\sqrt{5}}{2} \approx 1.618$ is the golden ratio
(see~\cite{Spirkl}). In particular, this implies that any
prefix-based \aop~realization has a depth of
at least $1.44 \log_2 m -1$.
Without using prefix gates,
Rautenbach et al.\ \cite{RautenbachEtal03} presented a
circuit for an {\aop } with delay at most
$(1 + \varepsilon) \left\lceil \log_2 W \right\rceil + c_\varepsilon$
(for any $\varepsilon > 0$), where $c_\varepsilon$ is a number
depending on $\varepsilon$, only.
Spirkl~\cite{SpirklMaster} specified the delay bound to
$(1 + \varepsilon) \left\lceil \log_2 W \right\rceil + \frac{6}{\varepsilon} + 8 + 5\varepsilon$
and improved it to
$(1 + \varepsilon) \left\lceil \log_2 W \right\rceil + \frac{3}{\varepsilon} + 5$.
Moreover, Spirkl~\cite{SpirklMaster} described a circuit with a delay
of at most
$\left\lceil \log_2 W \right\rceil + 2 \sqrt{2 \log_2 m - 2} + 6$.
Note that for any $\varepsilon > 0$, this is actually a better delay bound than
$(1 + \varepsilon) \left\lceil \log_2 W \right\rceil + \frac{3}{\varepsilon} + 5$
(because
$\varepsilon \log_2 W +  \frac{3}{\varepsilon} + 5
\geq 2\sqrt{3 \log_2 W} + 5
\geq 2\sqrt{3 \log_2 m} + 5 \geq
2 \sqrt{2 \log_2 m - 2} + 6$).
Up to now, this is the fastest circuit for \aop s with non-uniform arrival times.
Table \ref{non-uniform bounds} summarizes these results in comparison with our delay bound.
We also state size (i.e., number of gates) and maximum fanout of the constructed circuits.
Note that some methods trade off size against fanout and provide two different circuits.

All the \aop{} circuits presented in the table can be used directly
to obtain adder circuits with the same delays by computing each carry bit separately.
However, the size of the arising circuits would be super-quadratic.
The construction of linear-size adder circuits with our delay guarantee is part of ongoing work.

\subsection{Our Contribution}

In this paper, we present an algorithm with running time $\mathcal O(m^2 \log_2 m)$
that computes a Boolean circuit for \aop s (with $m \geq 3$)
using only two-input \AND~and \OR~gates
with a delay of at most
\[
   \log_2 W + \log_2 \log_2 m + \log_2 \log_2 \log_2 m + 4.3\,,
\]
size $\mathcal O(m \log_2 m)$
and maximum fanout $\log_2 m + \log_2 \log_2 m + \log_2 \log_2 \log_2 m + 3.3$.
In terms of delay, this yields the currently best known circuits for \aop s and thus also adder circuits.
In particular, we improve the previously best known delay bound
of $\left\lceil \log_2 W \right\rceil + 2 \sqrt{2 \log_2 m - 2} + 6$ by Spirkl \cite{SpirklMaster}
for each $m \geq 3$ as well as asymptotically.
The construction of the circuit is based on a recursive approach similar to the algorithm
of Grinchuk~\cite{Grinchuk} for uniform arrival times.



The rest of the paper is organized as follows.
In Section~\ref{sec::preliminaries}, we introduce basic definitions and results.
We give a formal description of the problem (Subsection~\ref{subsec:problem}),
define splitting steps which allow us to partition an instance into
smaller sub-instances (Subsection~\ref{subsec::recursion_formulas_ats})
and introduce a measure for deciding which instances
admit an \aop~realization with a given delay (Subsection \ref{subsec::delay:weight}).
Section~\ref{sec::bound_v} classifies these instances, which is the major step of the paper.
In Section~\ref{sec::bound_delay},
we deduce how to construct circuits realizing \aop s with a delay of at most
$\log_2 W + \log_2 \log_2 m + \log_2 \log_2 \log_2 m + 4.3$
and analyze the size and fanout as well as the runtime needed to compute such circuits.
In Section~\ref{generalized aops}, we extend this result to paths that also
consist of \AND~and \OR~gates only, but not necessarily alternatingly.

\section{Preliminaries}\label{sec::preliminaries}

\subsection{Problem Formulation} \label{subsec:problem}

We denote the set of natural numbers including zero by $\N$.
Our notation regarding Boolean functions and circuits is based on Savage \cite{Savage:1997:MCE:522806}.
Given $r \in \N$ and a Boolean function $h \colon \{0, 1\}^r \to \{0, 1\}$
with Boolean input variables (shorter, inputs) $x_0, \dotsc, x_{r-1}$,
we write $x = (x_0, \dotsc, x_{r-1})$ as a shorthand for all inputs with fixed ordering.
If $r = 0$, we write $x = ()$.

\begin{definition}\label{aop}
   Let Boolean input variables $t = (t_0, \dotsc, t_{m-1})$ for some $m \in \N$ with $m > 0$ be given.
   We call each of the recursively defined functions
   \[g(t) = \begin{cases}
                    t_0 & m = 1 \\
                    t_0 \land g^*((t_1, \dotsc, t_{m-1})) & m > 1
                 \end{cases} \quad \text{ and } \quad
    g^*(t) = \begin{cases}
                    t_0 & m = 1 \\
                    t_0 \lor g((t_1, \dotsc, t_{m-1})) & m > 1
                 \end{cases} \]
   an \emph{\aop} on $m$ inputs.
\end{definition}

\input{ext_aop_figure.tex}

We want to realize a given {\aop} as a Boolean circuit over the basis $\{\land, \lor\}$.
A Boolean circuit over the basis $\{\land, \lor\}$
is a directed acyclic graph such that
\begin{itemize}
 \item the nodes with indegree 0, called inputs, are labeled by Boolean variables,
 \item there is only one node with outdegree 0, called output, and it has indegree exactly 1,
 \item each of the remaining nodes has indegree exactly 2 and outdegree at least 1 and is labeled
       either as an \AND~gate or as an \OR~gate.
\end{itemize}
The logical function on the input variables computed by the output of the circuit can be obtained by
combining the logical functions represented by the gates recursively.
A circuit is a realization of the
{\aop} $g(t)$ if and only if the output signal equals $g(t)$ for all input values
(correspondingly for $g^*$).
Figure~\ref{ext aop figure}~(a) shows a Boolean circuit which is a
straightforward realization of the \aop~$g^*((t_0, \dotsc, t_5))$.
We omit drawing the output in such pictures
-- the unique predecessor of the output will always be the bottommost gate.

We assume that each input variable is associated with a prescribed arrival time.

\begin{definition}
 Let $r \in \N$ and
 a Boolean function $h \colon \{0, 1\}^r \to \{0, 1\}$
 on Boolean input variables $x = (x_0, \dotsc, x_{r-1})$
 with \emph{arrival times} $a \colon \{x_0, \dotsc, x_{r-1}\} \to \N$ be given.
 Consider a circuit $C$ computing $h$.
 For $i = 0, \dotsc, r-1$, the \emph{delay of input} $x_i$ is defined as
 $a(x_i) + l(x_i)$, where $l(x_i)$ denotes the maximum number of gates
 of any directed path in $C$ starting at input $x_i$.
 The \emph{delay of the circuit} $C$ is the maximum delay of any input.
\end{definition}

Note that for uniform arrival times, e.g., $a \equiv 0$, the delay of a circuit
is simply the depth of the circuit.

Our goal is computing fast circuits for \aop s, i.e., solving the following problem:

\begin{problem}[\AOP~Optimization Problem]
For $m \in \N$ with $m > 0$,
consider Boolean input variables $t = (t_0, \dotsc, t_{m-1})$
with arrival times $a \colon \{t_0, \dotsc, t_{m-1}\} \to \N$.
Find circuits over the basis $\{\land, \lor\}$
that compute the \aop s $g(t)$ and $g^*(t)$ with minimum possible delay.
\end{problem}

Note that adding the same constant to all arrival times does not change
the problem. This is why we only allow non-negative arrival times in this formulation.
Moreover, forbidding non-integral arrival times is not a significant restriction
because rounding
up all arrival times will change the delay of any circuit by less than $1$.
Therefore, we only consider natural numbers as arrival times.

\subsection{Recursive Circuit Construction}\label{subsec::recursion_formulas_ats}

We construct fast circuits for {\aop }s in a recursive way.
Before describing all details
of the approach, we explain the idea of the induction step.

Suppose we want to realize the {\aop }
$g^*(t) = t_{0} \lor (t_{1} \land t_{2} \lor (t_{3} \land ( \dotsc t_{m-1} \dotsc)))$.
We subdivide the inputs $t_0, \dotsc, t_{m-1}$ into
two groups  $t_0, \dotsc, t_{2k}$ and  $t_{2k+1}, \dotsc, t_{m-1}$.
Recursively, we compute fast circuits for the {\aop s} on each of these input sets.
These two circuits can be combined to
a circuit for the whole {\aop} as illustrated in
an example with $m = 12$ and $k = 3$ in Figure~\ref{alternating split figure};
the general construction is described in Lemma~\ref{alternating split proof}.
The output of the circuit for the subinstance $t_{2k+1}, \dotsc, t_{m-1}$
is combined with every second input of $t_0, \dotsc, t_{2k}$ by using
only {\AND } gates.
Just one additional \OR~gate (labeled ``$G$'' in the picture)
is needed to compute a circuit for the whole {\aop}. It is not too
difficult to check that the circuits in Figure~\ref{alternating split figure}~(a)
and (b) are logically equivalent.
To see this, note that the output of the circuit in (a) is ``true'' if and only if
there is an $i \in \{0,2,4,6,8,10,11\}$ such that the input signals
at $t_i$ and at every $t_j$ with $j$ odd and
$j \in \{1,\dots, i - 1\}$ are ``true''. This is also a sufficient
and necessary condition for a ``true'' as an output of the circuit in (b).

Note that while the left input of gate $G$ in the example is the output
of an {\aop}, the right input of $G$ is the output of a
function combining the {\aop } for $t_{2k+1}, \dotsc, t_{m-1}$
with a multiple-input {\sc AND}.
The occurrence of such functions in our recursion requires generalizing the concept of
\aop s.

\begin{definition} \label{ext aop}
   Given $n, m \in \N$, $m > 0$, and inputs $s_0, \dotsc, s_{n-1}, t_0, \dotsc, t_{m-1}$
   subdivided into $s = (s_0, \dotsc, s_{n-1})$ and $t = (t_0, \dotsc, t_{m-1})$,
   we define the \emph{extended \aop s}
   \[ f(s, t) = s_0 \land f((s_1, \dotsc, s_{n-1}), t)
   \quad \text{ and } \quad
   f^*(s, t) = s_0 \lor f^*((s_1, \dotsc, s_{n-1}), t) \,, \]
   where $f(s, t) = g(t)$, $f^*(s, t) = g^*(t)$ in the case that $s = ()$.
   We call the input variables $s$ \emph{symmetric inputs}
   and the input variables $t$ \emph{alternating inputs}, respectively.
\end{definition}

Figure~\ref{ext aop figure}~(b) shows the extended \aop~$f(s, t)$
for $s = (s_0, s_1, s_2)$ and $t = (t_0, \dotsc, t_4)$.
We shall always assume that the set of input variables contained in $s$ and $t$ are disjoint sets
indexed by $s_0, \dotsc, s_{n-1}$ and $t_0, \dotsc, t_{m-1}$.
Note that expanding the definitions of $f(s, t)$ and $f^*(s, t)$ given in Definition~\ref{ext aop} yields
\begin{align*}
f(s, t) &= s_0 \land \dotsc \land s_{n - 1} \land g(t) \hspace{-.65cm} &= s_0 \land \dotsc \land s_{n - 1} \land
                                      t_{0} \land (t_{1} \lor t_{2} \land (t_{3} \lor ( \dotsc t_{m-1} \dotsc)))\,,\\
f^*(s, t) &= s_0 \lor \dotsc \lor s_{n - 1} \lor g^*(t) \hspace{-.65cm} &= s_0 \lor \dotsc \lor s_{n - 1} \lor
                                        t_{0} \lor (t_{1} \land t_{2} \lor(t_{3} \land ( \dotsc t_{m-1} \dotsc)))\,,
\end{align*}
where, for $m$ odd, the innermost operation of $f(s, t)$ and $f^*(s, t)$ is $\lor$
and $\land$, respectively, and vice versa for $m$ even.

Due to the duality principle of Boolean algebra,
over the basis $\{\land, \lor\}$,
any realization for $f(s, t)$
yields a realization of $f^*(s, t)$ and vice
versa by switching all \AND~and \OR~gates.
In order to compute fast realizations for $f$ and $f^*$,
we will apply two methods that allow realizing $f$ and $f^*$ recursively.
Each of these methods reduces the problem of realizing $f(s, t)$ to the problem
of realizing extended \aop s with strictly fewer symmetric or alternating inputs.

First, we formally describe the well-known method which is depicted in Figure~\ref{alternating split figure}
for the special case that $s = ()$, $m = 12$ and $k = 3$.

\input{split_figure.tex}

\begin{lemma} \label{alternating split proof}
Let input variables $s$ and $t$ and an integer $k$ with $0 \leq k < \frac{m-1}{2}$ be given.
Denote by $t'$ the (odd-length) prefix $t' = (t_0, t_1, \dotsc, t_{2k})$ of $t$,
and by $t''$ the remaining inputs of $t$, i.e., $t'' = (t_{2k+1}, \dotsc, t_{m-1})$.
Then, we have
 \begin{equation}
f^*(s, t) = f^*(s, t') \lor f(\widehat{t'}, t'')\,, \label{alternating split}
\end{equation}
where $\widehat{t'} := (t'_1, t'_3, t'_5, \dotsc, t'_{2k-1})$ contains every second entry of $t'$.
\proof
 At first, we prove the following claim:

 \begin{claim}
 For any Boolean variable $x$, we have
 \begin{equation} \label{split proof}
   g^*((t_0, \dotsc, t_{2k}, x)) = g^*((t_0, \dotsc, t_{2k})) \lor (t_1 \land t_3 \land \dotsc \land t_{2k-1} \land x)\,.
 \end{equation}
 \begin{proof_of_claim}
  We prove the claim by induction on $k$.
  For $k = 0$, we have
  \[g^*((t_0, \dotsc, t_{2k}, x)) = g^*((t_0, x)) = t_0 \lor x = g^*((t_0, \dotsc, t_{2k})) \lor (t_1 \land t_3 \land \dotsc \land t_{2k-1} \land x)\,.\]
  Assuming statement (\ref{split proof}) holds for some $k \geq 0$ and setting $x_1 := t_{2k+1} \land (t_{2k+2} \lor x)$
  and $x_2 := t_{2k+1} \land t_{2k+2}$, we compute
  { 
  \newcommand{\ca}[1]{&\customalign{======}{#1}}
  \begin{align*}
   \ca{} g^*((t_0, \dotsc, t_{2(k+1)}, x)) \\
   \ca{\stackrel{\text{Def. }\ref{aop}}{=}} g^*((t_0, \dotsc, t_{2k}, x_1)) \\
   \ca{\stackrel{(\ref{split proof})}{=}} g^*((t_0, \dotsc, t_{2k})) \lor (t_1 \land t_3 \land\dotsc \land t_{2k-1} \land x_1) \\
   \ca{=} g^*((t_0, \dotsc, t_{2k})) \lor (t_1 \land t_3 \land \dotsc \land t_{2k-1} \land t_{2k+1} \land (t_{2k+2} \lor x)) \\
   \ca{=} g^*((t_0, \dotsc, t_{2k})) \lor (t_1 \land  t_3 \land\dotsc \land t_{2k-1} \land x_2) \lor (t_1 \land t_3 \land\dotsc \land t_{2k-1} \land t_{2k+1} \land x) \\
   \ca{\stackrel{(\ref{split proof})}{=}} g^*((t_0, \dotsc, t_{2k}, x_2)) \lor (t_1 \land t_3 \land \dotsc \land t_{2k-1} \land t_{2k+1} \land x) \\
   \ca{\stackrel{\text{Def. }\ref{aop}}{=}}  g^*((t_0, \dotsc, t_{2k+2})) \lor (t_1 \land t_3 \land \dotsc \land t_{2k-1} \land t_{2k+1} \land x)\,.
  \end{align*}
  }
  This proves the inductive step and thus the claim.
  \end{proof_of_claim}
  \end{claim}

 From this claim and Definition \ref{aop}, we conclude
 { 
  \newcommand{\ca}[1]{&\customalign{======}{#1}}
 \begin{align*}
  f^*(s, t) \ca{\stackrel{\text{Def. }\ref{ext aop}}{=}} s_0 \lor \dotsc \lor s_{n-1} \lor g^*(t) \\
            \ca{\stackrel{\text{Def. }\ref{aop}}{=}} s_0 \lor \dotsc \lor s_{n-1} \lor g^*((t_0, \dotsc, t_{2k}, g(t''))) \\
            \ca{\stackrel{(\ref{split proof})}{=}} s_0 \lor \dotsc \lor s_{n-1} \lor g^*((t_0, \dotsc, t_{2k})) \lor (t_1 \land t_3 \land \dotsc \land t_{2k-1} \land g(t'')) \\
            \ca{\stackrel{\text{Def. }\ref{ext aop}}{=}} f^*(s, t') \lor f(\widehat{t'}, t'')\,. \qedhere
  \end{align*}
  }
\end{lemma}

The previous lemma allows us to recursively construct a circuit for $f^*(s, t)$:
Inductively, we can assume that
we can compute a fast circuit for $f^*(s, t')$ and $f(\widehat{t'}, t'')$, respectively.
Combining these circuits as in equation~(\ref{alternating split}) yields a good circuit for $f^*(s, t)$.
We call this way to construct a circuit for $f^*(s, t)$ an \emph{alternating split}.

Similarly, if $n > 0$, the definition of $f(s, t)$ (see Definition \ref{ext aop})
yields two \emph{symmetric splits} indicated by the equations
\begin{align}
f(s, t) &= (s_0 \land \dotsc \land s_{n - 1}) \land g(t)\,, \label{symmetric split A} \\
f(s, t) &= (s_0 \land \dotsc \land s_{n - 1} \land t_0) \land g^*((t_1, \dotsc, t_{m-1}))\,. \label{symmetric split B}
\end{align}
Here, we only inquire a recursive circuit construction for the arising \aop~since
a delay-optimum symmetric \AND-tree can be constructed via Huffman coding \cite{huf52}
as explained in Remark~\ref{sym tree bound}.

Note that dualizing the splits (\ref{alternating split}), (\ref{symmetric split A}) and (\ref{symmetric split B})
yields analogous splits for $f^*$:
\begin{align}
f(s, t)   &= f(s, t') \land f^*(\widehat{t'}, t'') \label{alternating split dual} \\
f^*(s, t) &= (s_0 \lor \dotsc \lor s_{n - 1}) \lor g^*(t) \label{symmetric split A dual} \\
f^*(s, t) &= (s_0 \lor \dotsc \lor s_{n - 1}  \lor t_0) \lor g((t_1, \dotsc, t_{m-1})) \label{symmetric split B dual}
\end{align}

\subsection{Delay and Weight}\label{subsec::delay:weight}

We will realize extended \aop s with good delay using the
recursive methods defined in Subsection \ref{subsec::recursion_formulas_ats}.
For this, we need to classify inputs by their weight.

\begin{definition} \label{weight def}
 Given inputs $x = (x_0, \dotsc, x_{r-1})$ and arrival times $a$,
 for $i = 0, \dotsc, r-1$, the \emph{weight} of input $x_i$ is $W(x_i) := 2^{a(x_i)}$,
 and the \emph{weight} of $x$ is $W(x) := \sum_{i = 0}^{r-1} W(x_i)$.
\end{definition}

Note that in this definition, the weights $W(x_i)$ and $W(x)$ do not only depend on the inputs,
but also on the input arrival times.

\begin{rem} \label{W >= 1}
 By Definition \ref{weight def}, we have $W(x_i) \geq 1$ for any input $x_i$ with $a(x_i) \geq 0$.
\end{rem}

The definition of the weight is motivated by the fact that
for symmetric binary functions, i.e., \AND~trees or \OR~trees,
the optimum delay achievable for any realization can be derived from the weight of the inputs directly.

\begin{rem}\label{sym tree bound}
 A binary tree on inputs $x=(x_0, \dotsc, x_{r-1})$ with weight $W(x)$
 can be realized with delay $d$ if and only if Kraft's inequality $W(x) \leq 2^d$ is fulfilled \cite{Kraft},
 see also Golumbic \cite{Golumbic:1976}.
 A delay-optimum tree can be computed in runtime $\mathcal O(r \log_2 r)$ via a greedy algorithm
 based on Huffman coding \cite{huf52}.
A short proof of this can be found in Werber \cite{werber2007logic}.
\end{rem}

Note that when $t$ has at most $2$ entries, $f(s, t)$ is a symmetric
binary \AND~tree.

\begin{lemma}\label{t = (t_0)}
 For Boolean input variables $s$ and $t$ with $m \leq 2$ and arrival times $a$,
 we can realize $f(s, t)$ with delay $d$ if and only if
 $W(s) + W(t) \leq 2^d$. \qed
\end{lemma}

For the case that $t$ has more than $2$ entries, we give an upper bound
on the delay of $f(s, t)$ in Section~\ref{sec::bound_v}.

%% file: ext_aop_figure.tex
\begin{figure}
\centering
\begin{minipage}{0.42\textwidth}
\centering
\vspace{-.1cm}
\resizebox{0.897\textwidth}{!}{%
\begin{tikzpicture}

\node[outer sep=0pt, scale = 1.6] (x_2) at (2.5, 3.2){$t_0$};
\node[outer sep=0pt, scale = 1.6] (x_3) at (3.5, 3.2){$t_1$};
\node[outer sep=0pt, scale = 1.6] (x_4) at (4.5, 3.2){$t_2$};
\node[outer sep=0pt, scale = 1.6] (x_5) at (5.5, 3.2){$t_3$};
\node[outer sep=0pt, scale = 1.6] (x_6) at (6.5, 3.2){$t_4$};
\node[outer sep=0pt, scale = 1.6] (x_7) at (7.5, 3.2){$t_5$};

\node[fill=green, outer sep=0pt, or gate US, draw, logic gate inputs=nn, rotate=270, thick] at (7,2) (or1){};
\draw[thick] (x_7) -- (or1.input 1);
\draw[thick] (x_6) -- (or1.input 2);

\node[fill=red, outer sep=0pt, and gate US, draw, logic gate inputs=nn, rotate=270, thick] at (6,1) (and1){};
\draw[thick] (or1.output) -- (and1.input 1);
\draw[thick] (x_5) -- (and1.input 2);

\node[fill=green, outer sep=0pt, or gate US, draw, logic gate inputs=nn, rotate=270, thick] at (5,0) (or2){};
\draw[thick] (and1.output) -- (or2.input 1);
\draw[thick] (x_4) -- (or2.input 2);

\node[fill=red, outer sep=0pt, and gate US, draw, logic gate inputs=nn, rotate=270, thick] at (4,-1) (and3){};
\draw[thick] (or2.output) -- (and3.input 1);
\draw[thick] (x_3) -- (and3.input 2);

\node[fill=green, outer sep=0pt, or gate US, draw, logic gate inputs=nn, rotate=270, thick] at (3,-2) (or4){};
\draw[thick] (and3.output) -- (or4.input 1);
\draw[thick] (x_2) -- (or4.input 2);

\phantom{\node[outer sep=0pt, scale = 1.6] (x_0) at (3.5, 3.2){$s_0$};}
\phantom{\node[outer sep=0pt, scale = 1.6] (x_{10}) at (10.5, 3.2){$t_4$};}
\phantom{\node[fill=yellow, outer sep=0pt, and gate US, draw, logic gate inputs=nn, rotate=270, thick] at (4,-4) (and5){};}

\end{tikzpicture}}
\caption*{(a)\, \aop~$g^*((t_0, \dotsc, t_{5}))$.}
\end{minipage}
\hfill
\begin{minipage}{0.57\textwidth}
\centering
\resizebox{0.583\textwidth}{!}{%
\begin{tikzpicture}

\node[outer sep=0pt, scale = 1.6] (x_0) at (3.5, 3.2){$s_0$};
\node[outer sep=0pt, scale = 1.6] (x_2) at (4.5, 3.2){$s_1$};
\node[outer sep=0pt, scale = 1.6] (x_4) at (5.5, 3.2){$s_2$};
\node[outer sep=0pt, scale = 1.6] (x_6) at (6.5, 3.2){$t_0$};
\node[outer sep=0pt, scale = 1.6] (x_7) at (7.5, 3.2){$t_1$};
\node[outer sep=0pt, scale = 1.6] (x_8) at (8.5, 3.2){$t_2$};
\node[outer sep=0pt, scale = 1.6] (x_9) at (9.5, 3.2){$t_3$};
\node[outer sep=0pt, scale = 1.6] (x_{10}) at (10.5, 3.2){$t_4$};

\node[fill=green, outer sep=0pt, or gate US, draw, logic gate inputs=nn, rotate=270, thick] at (10,2) (or0){};
\draw[thick] (x_{10}) -- (or0.input 1);
\draw[thick] (x_9) -- (or0.input 2);

\node[fill=red, outer sep=0pt, and gate US, draw, logic gate inputs=nn, rotate=270, thick] at (9,1) (and1){};
\draw[thick] (or0.output) -- (and1.input 1);
\draw[thick] (x_8) -- (and1.input 2);

\node[fill=green, outer sep=0pt, or gate US, draw, logic gate inputs=nn, rotate=270, thick] at (8,0) (or1){};
\draw[thick] (and1.output) -- (or1.input 1);
\draw[thick] (x_7) -- (or1.input 2);

\node[fill=red, outer sep=0pt, and gate US, draw, logic gate inputs=nn, rotate=270, thick] at (7,-1) (and2){};
\draw[thick] (or1.output) -- (and2.input 1);
\draw[thick] (x_6) -- (and2.input 2);

\node[fill=yellow, outer sep=0pt, and gate US, draw, logic gate inputs=nn, rotate=270, thick] at (6,-2) (and3){};
\draw[thick] (and2.output) -- (and3.input 1);
\draw[thick] (x_4) -- (and3.input 2);

\node[fill=yellow, outer sep=0pt, and gate US, draw, logic gate inputs=nn, rotate=270, thick] at (5,-3) (and4){};
\draw[thick] (and3.output) -- (and4.input 1);
\draw[thick] (x_2) -- (and4.input 2);

\node[fill=yellow, outer sep=0pt, and gate US, draw, logic gate inputs=nn, rotate=270, thick] at (4,-4) (and5){};
\draw[thick] (and4.output) -- (and5.input 1);
\draw[thick] (x_0) -- (and5.input 2);

\end{tikzpicture}}
\caption*{(b)\, Extended~\aop~$f((s_0, s_1, s_2), (t_0, \dotsc, t_4))$.}
\end{minipage}
\caption{Examples for (extended) \aop s.}
\label{ext aop figure}
\end{figure}

%% file: split_figure.tex
\begin{figure}
\centering
\providecommand{\subfigw}{0.44\textwidth}
\begin{minipage}{\subfigw}
\resizebox{\textwidth}{!}{%
\begin{tikzpicture}
\node[outer sep=0pt, scale = 1.6] (x0)  at (-1.5, 6.2){$t_0$};
\node[outer sep=0pt, scale = 1.6] (x1)  at (-0.5, 6.2){$t_1$};
\node[outer sep=0pt, scale = 1.6] (x2)  at (0.5,  6.2){$t_2$};
\node[outer sep=0pt, scale = 1.6] (x3)  at (1.5,  6.2){$t_3$};
\node[outer sep=0pt, scale = 1.6] (x4)  at (2.5,  6.2){$t_4$};
\node[outer sep=0pt, scale = 1.6] (x5)  at (3.5,  6.2){$t_5$};
\node[outer sep=0pt, scale = 1.6] (x6)  at (4.5,  6.2){$t_6$};
\node[outer sep=0pt, scale = 1.6] (x7)  at (5.5,  6.2){$t_7$};
\node[outer sep=0pt, scale = 1.6] (x8)  at (6.5,  6.2){$t_8$};
\node[outer sep=0pt, scale = 1.6] (x9)  at (7.5,  6.2){$t_9$};
\node[outer sep=0pt, scale = 1.6] (x10) at (8.5,  6.2){$t_{10}$};
\node[outer sep=0pt, scale = 1.6] (x11) at (9.5,  6.2){$t_{11}$};

\draw[thick, decorate,decoration={brace, amplitude=5pt, raise=10pt}] (x0.west) -- (x6.east) node [midway, above, sloped, yshift=15pt, scale = 1.6] {$t'$};
\draw[thick, decorate,decoration={brace, amplitude=5pt, raise=10pt}] (x7.west) -- (x11.east) node [midway, above, sloped, yshift=15pt, scale = 1.6] {$t''$};

\node[fill=green, outer sep=0pt, or gate US, draw, logic gate inputs=nn, rotate=270, thick] at (7,3) (or1){};
\node[fill=red, outer sep=0pt, and gate US, draw, logic gate inputs=nn, rotate=270, thick] at (6,2) (and1){};

\node[fill=green, outer sep=0pt, or gate US, draw, logic gate inputs=nn, rotate=270, thick] at (5,1) (or4){};
\node[fill=red, outer sep=0pt, and gate US, draw, logic gate inputs=nn, rotate=270, thick] at (4,0) (and2){};

\node[fill=green, outer sep=0pt, or gate US, draw, logic gate inputs=nn, rotate=270, thick] at (3,-1) (or5){};

\node[fill=red, outer sep=0pt, and gate US, draw, logic gate inputs=nn, rotate=270, thick] at (8,4) (and7){};
\node[fill=green, outer sep=0pt, or gate US, draw, logic gate inputs=nn, rotate=270, thick] at (9,5) (or8){};

\node[fill=red, outer sep=0pt, and gate US, draw, logic gate inputs=nn, rotate=270, thick] at (2,-2) (and6){};
\node[fill=green, outer sep=0pt, or gate US, draw, logic gate inputs=nn, rotate=270, thick] at (1,-3) (or9){};

\node[fill=red, outer sep=0pt, and gate US, draw, logic gate inputs=nn, rotate=270, thick] at (0,-4) (and10){};
\node[fill=green, outer sep=0pt, or gate US, draw, logic gate inputs=nn, rotate=270, thick] at (-1,-5) (or11){};

\draw[thick, cyan] ($(and1.output) + (-0.8, 0.3)$) -- ($(and1.output) - (0, 0.5)$);

\draw[thick] (or1.output) -- (and1.input 1);
\draw[thick] (x8) -- (or1.input 2);
\draw[thick] (x7) -- (and1.input 2);
\draw[thick] (x9) -- (and7.input 2);

\draw[thick] (or4.output) -- (and2.input 1);
\draw[thick] (or8.output) -- (and7.input 1);

\draw[thick] (x5) -- (and2.input 2);
\draw[thick] (x10) -- (or8.input 2);
\draw[thick] (x11) -- (or8.input 1);

\draw[thick] (x6) -- (or4.input 2);
\draw[thick] (and1.output) -- (or4.input 1);

\draw[thick] (and2.output) -- (or5.input 1);
\draw[thick] (or5.output) -- (and6.input 1);

\draw[thick] (x4) -- (or5.input 2);
\draw[thick] (x3) -- (and6.input 2);

\draw[thick] (x2) -- (or9.input 2);
\draw[thick] (and6.output) -- (or9.input 1);
\draw[thick] (and7.output) -- (or1.input 1);

\draw[thick] (or9.output) -- (and10.input 1);
\draw[thick] (x1) -- (and10.input 2);
\draw[thick] (and10.output) -- (or11.input 1);
\draw[thick] (x0) -- (or11.input 2);

\end{tikzpicture}}
\caption*{(a) \, \aop~on $t = (t_0, \dotsc, t_{11})$.}
\end{minipage}
\hfill
\begin{minipage}{\subfigw}
\resizebox{\textwidth}{!}{%
\begin{tikzpicture}

\node[outer sep=0pt, scale = 1.6] (x0) at (-3.5, 6.2){$t_0$};
\node[outer sep=0pt, scale = 1.6] (x1) at (-2.5, 6.2){$t_1$};
\node[outer sep=0pt, scale = 1.6] (x2) at (-1.5, 6.2){$t_2$};
\node[outer sep=0pt, scale = 1.6] (x3) at (-0.5, 6.2){$t_3$};
\node[outer sep=0pt, scale = 1.6] (x4) at (0.5, 6.2){$t_4$};
\node[outer sep=0pt, scale = 1.6] (x5) at (1.5, 6.2){$t_5$};
\node[outer sep=0pt, scale = 1.6] (x6) at (2.5, 6.2){$t_6$};
\node[outer sep=0pt, scale = 1.6] (x7) at (3.5, 6.2){$t_7$};
\node[outer sep=0pt, scale = 1.6] (x8) at (4.5, 6.2){$t_8$};
\node[outer sep=0pt, scale = 1.6] (x9) at (5.5, 6.2){$t_9$};
\node[outer sep=0pt, scale = 1.6] (x10) at (6.5, 6.2){$t_{10}$};
\node[outer sep=0pt, scale = 1.6] (x11) at (7.5, 6.2){$t_{11}$};
\node[fill=green, outer sep=0pt, or gate US, draw, logic gate inputs=nn, rotate=270, thick, opacity=0] at (-1,-5) (pseudo){};

\draw[thick, decorate,decoration={brace, amplitude=5pt, raise=10pt}] (x0.west) -- (x6.east) node [midway, above, sloped, yshift=15pt, scale = 1.6] {$t'$};
\draw[thick, decorate,decoration={brace, amplitude=5pt, raise=10pt}] (x7.west) -- (x11.east) node [midway, above, sloped, yshift=15pt, scale = 1.6] {$t''$};

\node[fill=red, outer sep=0pt, and gate US, draw, logic gate inputs=nn, rotate=270, thick] at (2,5) (and1){};
\draw[thick] (x6) -- (and1.input 1);
\draw[thick] (x5) -- (and1.input 2);

\node[fill=green, outer sep=0pt, or gate US, draw, logic gate inputs=nn, rotate=270, thick] at (1.5,4) (or1){};
\draw[thick] (and1.output) -- (or1.input 1);
\draw[thick] (x4) -- (or1.input 2);

\node[fill=red, outer sep=0pt, and gate US, draw, logic gate inputs=nn, rotate=270, thick] at (1,3) (and2){};
\draw[thick] (or1.output) -- (and2.input 1);
\draw[thick] (x3) -- (and2.input 2);

\node[fill=green, outer sep=0pt, or gate US, draw, logic gate inputs=nn, rotate=270, thick] at (.5,2) (or9){};
\draw[thick] (and2.output) -- (or9.input 1);
\draw[thick] (x2) -- (or9.input 2);

\node[fill=red, outer sep=0pt, and gate US, draw, logic gate inputs=nn, rotate=270, thick] at (0,1) (and10){};
\draw[thick] (or9.output) -- (and10.input 1);
\draw[thick] (x1) -- (and10.input 2);

\node[fill=green, outer sep=0pt, or gate US, draw, logic gate inputs=nn, rotate=270, thick] at (-.5,0) (or11){};
\draw[thick] (and10.output) -- (or11.input 1);
\draw[thick] (x0) -- (or11.input 2);

\node[fill=green, outer sep=0pt, or gate US, draw, logic gate inputs=nn, rotate=270, thick] at (7,5) (or6){};
\draw[thick] (x11) -- (or6.input 1);
\draw[thick] (x10) -- (or6.input 2);

\node[fill=red, outer sep=0pt, and gate US, draw, logic gate inputs=nn, rotate=270, thick] at (6.5,4) (and2){};
\draw[thick] (or6.output) -- (and2.input 1);
\draw[thick] (x9) -- (and2.input 2);

\node[fill=green, outer sep=0pt, or gate US, draw, logic gate inputs=nn, rotate=270, thick] at (6,3) (or3){};
\draw[thick] (and2.output) -- (or3.input 1);
\draw[thick] (x8) -- (or3.input 2);

\node[fill=red, outer sep=0pt, and gate US, draw, logic gate inputs=nn, rotate=270, thick] at (5.5,2) (and9){};
\draw[thick] (or3.output) -- (and9.input 1);
\draw[thick] (x7) -- (and9.input 2);

\node[fill=yellow, outer sep=0pt, and gate US, draw, logic gate inputs=nn, rotate=270, thick] at (4.9,1) (and8){};
\draw[thick]       (and9.output) -- (and8.input 1);
\draw[thick] (x5) -- (and8.input 2);

\node[fill=yellow, outer sep=0pt, and gate US, draw, logic gate inputs=nn, rotate=270, thick] at (4.3,0) (and5){};
\draw[thick] (and8.output) -- (and5.input 1);
\draw[thick] (x3) -- (and5.input 2);

\node[fill=yellow, outer sep=0pt, and gate US, draw, logic gate inputs=nn, rotate=270, thick] at (3.7,-1) (and12){};
\draw[thick] (and5.output) -- (and12.input 1);
\draw[thick] (x1) -- (and12.input 2);

\node[fill=cyan, outer sep=0pt, or gate US, draw, logic gate inputs=nn, rotate=270, thick] at (1.5,-2) (or4){};
\draw[thick] (and12.output) -- (or4.input 1);
\draw[thick] (or11.output) -- (or4.input 2);
\node[outer sep=0pt] (x10) at (2.1, -2.1){$G$};

\draw[thick, cyan] ($(and9.output) + (-0.8, 0.3)$) -- ($(and9.output) - (-0.2, 0.5)$);

\end{tikzpicture}}
\caption*{(b) \, Alternating split with $t' = (t_0, \dotsc, t_6)$.}
\end{minipage}
\caption{Performing the alternating split on $g^*((t_0, \dotsc, t_{11})) = t_0 \lor (t_1 \land ( \dotsc (t_{10} \lor t_{11}) \dotsc ))$.}
\label{alternating split figure}
\end{figure}

%% file: bounding_V.tex
\section{Bounding the Weight for Given Delay}\label{sec::bound_v}

We will prove an upper bound on the delay of \aop s by a reverse argument
similarly as in Grinchuk's proof for the case of uniform input arrival times \cite{Grinchuk}.
Grinchuk fixes a depth bound $d$ and the number $n$ of symmetric inputs $s$,
and determines how many alternating inputs $t$ an \aop~may have such that $f(s, t)$
can be realized with depth $d$.
Similarly, given symmetric inputs $s$ with a fixed weight $w$ and a fixed delay bound $d$,
we will determine for which alternating inputs $t$ a realization for $f(s, t)$ with delay $d$ can be guaranteed.
Since it is difficult to classify these $t$ exactly,
we distinguish different alternating inputs $t$ by their weight only.

In order facilitate the formulation of our main statement,
we fix a constant $\Ceta$, see also Remark~\ref{rem:ceta}.

\begin{definition}
Let $\Ceta := 1.9$ be a fixed constant.
\end{definition}

We aim at proving the following statement:

\begin{theorem} \label{bound on v(d, w)}
 Let $d, w \in \N$ with $d > 1$ and $0 \leq w < 2^{d-1}$ be given.
 Consider Boolean input variables $s$ and $t$ with $W(s) = w$ and
 \[W(t) \leq \Ceta \frac{2^{d-1} - w}{d \log_2 (d)}\,.\]
 Then, there is a circuit realizing $f(s, t)$ with delay at most $d$.
\end{theorem}

We will prove Theorem~\ref{bound on v(d, w)} by induction on $d$.
Based on this, we will deduce the desired
upper bound of $\log_2 W + \log_2 \log_2 m + \log_2 \log_2 \log_2 m + 4.3$ on the delay of \aop s
in Section~\ref{sec::bound_delay}.

\begin{rem} \label{rem:ceta}
Note that the choice of the constant $\Ceta$ influences the additive constant (here $4.3$) in the delay bound.
For example, if we chose $\Ceta := 1$, we would need to replace the additive constant by $5$.
On the other hand, the higher $\Ceta$,
the larger the minimum value of $d$ is for which the induction step works,
which increases the proof complexity for the base case.

Most parts of the proof of Theorem~\ref{bound on v(d, w)} will work for any $\Ceta$ with $1 \leq \Ceta < 2$;
only at the end of the proof of Lemma~\ref{small w} we demand $\Ceta \leq 1.9$.
A slightly larger choice of $\Ceta$ would be possible, but using $\Ceta = 1.9$ keeps calculations simpler.
An improvement of the additive term to, e.g., $4.2$ would only be possible for $\Ceta \geq 1.992$,
for which we cannot prove Theorem~\ref{bound on v(d, w)}.
\end{rem}

For proving Theorem~\ref{bound on v(d, w)},
we would like to proceed by induction on $d$ making use of the restructuring
formulas presented in Section \ref{subsec::recursion_formulas_ats}.
The following remark explains why the inductive step from $d$ to $d+1$ does not work directly.

\begin{rem} \label{Grinchuk vs us}
Suppose that Theorem~\ref{bound on v(d, w)} holds for some $d > 1$ and all $0 \leq w < 2^{d-1}$.
In order to prove Theorem~\ref{bound on v(d, w)} for $d+1$,
we need to show that, given input variables $s$ and $t$ with $0 \leq w := W(s) < 2^d$ and $W(t) \leq \Ceta \frac{2^{d} - w}{(d+1) \log_2 (d+1)}$,
$f(s, t)$ can be realized with delay $d+1$.
In the case that $w < 2^{d-1}$, we would like to apply the alternating split
$f(s, t) = f(s, t') \land f^*(\widehat {t'}, t'')$
given by (\ref{alternating split dual}).
If we choose the prefix $t'$ of $t$ such that
\begin{equation} \label{W(t')}
W(t') \leq \Ceta \frac{2^{d-1} - w}{d \log_2 (d)}
\end{equation}
(whenever this is possible),
the induction hypothesis and the assumption $w < 2^{d-1}$ allow us to construct a circuit for $f(s, t')$ with delay $d$.
Thus, in order to construct a circuit with delay $d+1$ for $f(s, t)$,
it remains to prove that $f^*(\widehat {t'}, t'')$ admits a circuit with delay $d$.
Again, by induction hypothesis, we need to show that $W(t'') \leq \Ceta \frac{2^{d-1} - W(\widehat{t'})}{d \log_2 (d)}$.
But the only thing we know about $W(t'')$ is that
\begin{equation} \label{W(t'')}
 W(t'') = W(t) - W(t')\,.
\end{equation}

Even if we choose the prefix $t'$ maximal with (\ref{W(t')}),
this will not give us a meaningful upper bound on $W(t'')$ since $W(t')$ might be arbitrarily small in comparison to $W(t)$.

Note that this is what distinguishes our proof from Grinchuk's \cite{Grinchuk}:
  For arrival times all $0$,
  choosing $t'$ maximal with~(\ref{W(t')}) works well, since then, by maximality, we have
  $W(t') > \Ceta \frac{2^{d-1} - w}{d \log_2(d)} - 2$,
  hence equation~(\ref{W(t'')}) yields $W(t'') = W(t) - W(t') < W(t) - \Ceta \frac{2^{d-1} - w}{d \log_2(d)} - 2$.
  It turns out that this upper bound on $W(t^{**})$ suffices
  to prove that $f^*(\widehat {t'}, t'')$ can be realized with delay $d$.

  When arbitrary arrival times are present, a different proof idea is needed.
\end{rem}

Thus, instead of proving Theorem~\ref{bound on v(d, w)} via induction on $d$,
we strengthen the induction hypothesis and prove the stronger Theorem~\ref{main theorem}.

\begin{definition}
 Let $m \in \N$ with $m > 0$.
 For inputs $t = (t_0, \dotsc, t_{m-1})$ with arrival times $a$,
 we denote by $\Lambda_t$ the weight of the last two (or fewer) entries of $t$,
 i.e.,
 \[\Lambda_t := \begin{cases}
                W(t_0),               & m = 1\,, \\
                W(t_{m-2}) + W(t_{m-1}), & m > 1\,.
               \end{cases}\]
\end{definition}

\begin{theorem}\label{main theorem}
 Let $d, w \in \N$ with $d > 1$ and $0 \leq w < 2^{d-1}$ be given.
 Consider Boolean input variables $s$ and $t$ with $W(s) = w$ and
 \begin{equation}
  W(t) \leq \Ceta \frac{2^{d-1} - w}{d \log_2 (d)} + \frac{d-1}{d} \Lambda_t\,. \label{main req}
 \end{equation}
 Then, there is a circuit realizing $f(s, t)$ with delay at most $d$.
\end{theorem}

Note that since $\Lambda_t \geq 0$, Theorem~\ref{main theorem} implies Theorem~\ref{bound on v(d, w)}.
The proof of Theorem~\ref{main theorem} is the most important part of this paper and covers the rest of this section.
First, we observe two ways to express requirement~(\ref{main req}) differently.

\begin{rem}
 Assuming the conditions of Theorem~\ref{main theorem}, the following statements are equivalent to
 requirement~(\ref{main req}):
\begin{align}
   \sum_{i = 0}^{m - 3} W(t_i) + \frac{\Lambda_t}{d} &\leq \Ceta \frac{2^{d-1} - w}{d\log_2(d)} \label{main req - lambda} \\
   d \cdot \sum_{i = 0}^{m - 3} W(t_i) + \Lambda_t &\leq \Ceta \frac{2^{d-1} - w}{\log_2(d)} \label{main req * d}
\end{align}
Statement~(\ref{main req - lambda}) can be obtained from requirement~(\ref{main req}) by
subtracting $\frac{d-1}{d} \Lambda_t$.
From this, multiplication with $d$ yields statement (\ref{main req * d}).
\end{rem}

Next, we give an upper bound on the weight $W(t) + w$.

\begin{lemma} \label{lambda lemma}
Assuming the conditions of Theorem \ref{main theorem}, we have
\begin{align}
   W(t) + w \leq \begin{cases}
                      2^{d-1} & \text{ if } d \geq 2^{\Ceta}\,, \\
                      \frac{2^d}{\log_2 (d)} & \text{ otherwise}\,.
                   \end{cases} \label{bound W(t) + w}
\end{align}
\proof
Using inequality~(\ref{main req * d}), we obtain
\begin{equation} \label{proof bound W(t) + w}
  W(t) + w \stackrel{(\ref{main req * d})}{\leq} \Ceta \frac{2^{d-1} - w}{\log_2(d)} + w
             = \frac{\Ceta 2^{d-1} + (\log_2(d) - \Ceta) w}{\log_2(d)}\,.
\end{equation}
If $d \geq 2^{\Ceta}$, the condition $w < 2^{d-1}$ yields
\[
  W(t) + w \stackrel{(\ref{proof bound W(t) + w})}{\leq }\frac{\Ceta 2^{d-1} + (\log_2(d) - \Ceta) w}{\log_2(d)}
    \leq \frac{2^{d-1} ( \Ceta + \log_2(d) - \Ceta)}{\log_2(d)} = 2^{d-1}\,.
\]
Otherwise, if $d < 2^{\Ceta}$, the condition $w \geq 0$ yields
\[
  W(t) + w \stackrel{(\ref{proof bound W(t) + w})}{\leq }\frac{\Ceta 2^{d-1} + (\log_2(d) - \Ceta) w}{\log_2(d)}
    \leq \frac{\Ceta 2^{d-1} }{\log_2(d)}
    \stackrel{\Ceta < 2}{<} \frac{2^{d} }{\log_2(d)}\,. \qedhere
\]
\end{lemma}

The equivalent requirements (\ref{main req}), (\ref{main req - lambda}) and (\ref{main req * d})
as well as Lemma \ref{lambda lemma}
will be used extensively when proving Theorem \ref{main theorem}.
For this proof, we proceed by induction on $d$.
In Lemma~\ref{base case}, we will show as a base case
that Theorem \ref{main theorem} holds for $d \leq 3$.
Then, in Theorem~\ref{main theorem d+1}, we will prove the inductive step:
Assuming that Theorem \ref{main theorem} holds for some $d \geq 3$
and all $0 \leq w < 2^{d-1}$,
we will prove the statement for $d+1$ and all $0 \leq w < 2^d$.

\begin{lemma}\label{base case}
Assuming the conditions of Theorem \ref{main theorem} for $d = 2, 3$,
we can realize $f(s, t)$ with delay $d$.
\begin{proof}
 First assume that $m \leq 2$.
 Recall that in this case, $f(s, t)$ is a symmetric binary tree.
 By inequality~(\ref{bound W(t) + w}), we know that $W(t) + w \leq 2^{d}$.
 Hence, by Remark \ref{sym tree bound}, we can realize $f(s, t)$ with delay $d$
 using Huffman coding.

 Now let $m \geq 3$.
 Requirement (\ref{main req - lambda}) yields
 \begin{equation}
  1 + \frac{2}{d}
  \stackrel{m \geq 3, \text{Rem. } \ref{W >= 1}}{\leq} \sum_{i = 0}^{m-3} W(t_i) + \frac{\Lambda_t}{d}
  \stackrel{(\ref{main req - lambda})}{\leq} \Ceta \frac{2^{d-1} - w}{d\log_2(d)}  \label{small d contr}
  \stackrel{\zeta < 2}{<} \frac{2^d - 2w}{d\log_2(d)}\,.
 \end{equation}
 For $d = 2$, this leads to the contradiction
 \[2 = 1 + \frac{2}{2}
 \stackrel{(\ref{small d contr})}{<} \frac{4 - 2w}{2 \cdot 1}
 \stackrel{w \geq 0}{\leq} 2\,,\]
 i.e., for $d = 2$, we always have $m \leq 2$ and have already proven the required statement.

 Similarly, if $d = 3$, we obtain
 \[\frac{5}{3} = 1 + \frac{2}{3}
 \stackrel{(\ref{small d contr})}{<} \frac{8 - 2w}{3 \cdot \log_2(3)}
 < \begin{cases}
    2 & w = 0\\
    \frac{4}{3} & w \geq 1
   \end{cases}\,.\]
 In the case that $w \geq 1$, this is a contradiction;
 and for $w = 0$, the only remaining case is $m = 3$ with $t_0 = t_1 = t_2 = 0$
 for which $f(s, t) = t_0 \land (t_1 \lor t_2)$ can obviously be constructed with delay $2 < 3$.
\end{proof}
\end{lemma}

\begin{theorem}\label{main theorem d+1}
 Assume inductively that for some $d \geq 3$ and all $0 \leq w < 2^{d-1}$,
 Theorem \ref{main theorem} holds.
 Then, for inputs $s$ and $t$ with $w := W(s)$ such that $0 \leq w < 2^d$ and
 \begin{equation}
   W(t) \leq \Ceta \frac{2^{d} - w}{(d+1)\log_2(d+1)} + \frac{d}{d+1} \Lambda_t\,, \label{main req d+1}
 \end{equation}
 we can realize $f(s, t)$ with delay $(d+1)$.
\end{theorem}

As a sub-calculation for the proof of this theorem, we need the following lemma.

\begin{lemma} \label{sub-calc}
In the situation of Theorem \ref{main theorem d+1}, we have
{ 
  \newcommand{\ca}[1]{&\customalign{====}{#1}}
\begin{align*}
 \ca{} \Ceta \frac{2^{d-1}}{d \log_2(d)}  +  \frac{d-1}{d} \Lambda_t  -  \Ceta \frac{2^{d} - w}{(d+1)\log_2(d+1)}  -  \frac{d}{d+1} \Lambda_t \\
 \ca{\geq} \Ceta \frac{2^{d-1}\log_2(d+1) - (2^d - w) \log_2(d)}{d \log_2(d) \log_2(d+1)}\,.
\end{align*}
} 
\proof
Using the bound on $\Lambda_t$ implied by inequality~(\ref{main req * d}), we calculate
 { 
  \newcommand{\ca}[1]{&\customalign{===}{#1}}
\begin{align*}
 \ca{}  \Ceta \frac{2^{d-1}}{d \log_2(d)}  +  \frac{d-1}{d} \Lambda_t  -  \Ceta \frac{2^{d} - w}{(d+1)\log_2(d+1)} - \frac{d}{d+1} \Lambda_t \\
  \ca{=} \Ceta \frac{2^{d-1}(d+1) \log_2(d+1) - (2^d - w) d \log_2(d)}{d(d+1) \log_2(d) \log_2(d+1)} - \frac{1}{d(d+1)} \Lambda_t \\
  \ca{\stackrel{(\ref{main req * d})}{\geq}}
         \Ceta \frac{2^{d-1}(d+1) \log_2(d+1) - (2^d - w) d \log_2(d)}{d(d+1) \log_2(d) \log_2(d+1)} - \frac{1}{d(d+1)} \Ceta \frac{2^d - w}{\log_2(d+1)} \\
  \ca{=} \Ceta \frac{2^{d-1}\log_2(d+1) - (2^d - w) \log_2(d)}{d \log_2(d) \log_2(d+1)}\,. \qedhere
  \end{align*}
  } 
\end{lemma}

This is the only ingredient needed to prove Theorem \ref{main theorem d+1} for the
case that $2^{d-1} \leq w < 2^d$.

\begin{lemma} \label{large w}
 Theorem \ref{main theorem d+1} holds for all $w$ satisfying $2^{d-1} \leq w < 2^d$.
 \begin{proof}
  The symmetric split (\ref{symmetric split A}) yields the realization
  \begin{equation}
   f(s, t) = (s_0 \land \dotsc \land s_{n - 1}) \land g(t)\,. \label{realization large w}
  \end{equation}
  Since $w < 2^{d}$,
  Remark \ref{sym tree bound} allows the construction of a symmetric tree on inputs $s$ with delay $d$.
  In order to show that $f((), t) = g(t)$ can be realized with delay $d$,
  by induction hypothesis, it suffices to show the second inequality in
  \[W(t) \stackrel{(\ref{main req d+1})}{\leq}
         \Ceta \frac{2^{d} - w}{(d+1)\log_2(d+1)}  +  \frac{d}{d+1} \Lambda_t
       \leq \Ceta \frac{2^{d-1}}{d \log_2(d)}  +  \frac{d-1}{d} \Lambda_t\,.\]
  Subtracting the left-hand side from the right-hand side, we prove this via
  { 
  \newcommand{\ca}[1]{&\customalign{======}{#1}}
  \begin{align*}
  \ca{} \Ceta \frac{2^{d-1}}{d \log_2(d)} + \frac{d-1}{d} \Lambda_t - \Ceta \frac{2^{d} - w}{(d+1)\log_2(d+1)} - \frac{d}{d+1} \Lambda_t \\
  \ca{\stackrel{\text{Lem.\ }\ref{sub-calc}}{\geq}} \Ceta \frac{2^{d-1}\log_2(d+1) - (2^d - w) \log_2(d)}{d \log_2(d) \log_2(d+1)} \\
  \ca{\stackrel{w \geq 2^{d-1}}{\geq}} \Ceta \frac{2^{d-1}(\log_2(d+1) - \log_2(d))}{d \log_2(d) \log_2(d+1)} \\
  \ca{\geq} 0\,.
  \end{align*}
  } 
  Hence, applying the symmetric split (\ref{realization large w}) yields a realization for $f(s, t)$
  with delay $d+1$.
 \end{proof}
\end{lemma}

In the case $0 \leq w < 2^{d-1}$, we need a bound on the logarithm of consecutive integers.

\begin{rem} \label{log remark}
For $d \geq 3$, we have $d \geq \ln(2)(d+1)$ and thus
\begin{equation}
   \log_2(d+1) - \log_2(d) = \frac{\ln(d+1) - \ln(d)}{\ln(2)}
                            = \int_d^{d+1} \frac{1}{\ln(2) x} dx
                           \geq \frac{1}{\ln(2)(d+1)}
                           \geq \frac{1}{d}\,. \label{log(d+1) greater}
\end{equation}
\end{rem}

Now we will prove Theorem \ref{main theorem d+1} for the case that $0 \leq w < 2^{d-1}$.
\begin{lemma} \label{small w}
 Theorem \ref{main theorem d+1} holds for each $w$ satisfying $0 \leq w < 2^{d-1}$.
  \begin{proof}
  We prove this lemma via a case distinction.
  In Case 2, we will consider a prefix $t'$ of the inputs $t$ with weight at most
  $\Ceta \frac{2^{d-1} - w}{d \log_2(d)}$ in order to proceed similarly as indicated in Remark~\ref{Grinchuk vs us}.
  If the weight of $t_0$ is already larger than this, such a prefix does not exist.
  We deal with this situation in Case 1.

  \textbf{Case 1:} Assume that
  \begin{equation}
   W(t_0) > \Ceta \frac{2^{d-1} - w}{d \log_2(d)}\,. \label{t_0 case}
  \end{equation}
  The symmetric split (\ref{symmetric split B}) yields
  \begin{equation}
   f(s, t) = (s_0 \land \dotsc \land s_{n-1} \land t_0) \land g^*((t_1, t_2, \dotsc, t_{m-1}))\,. \label{case 1 realization}
  \end{equation}
  Due to inequality~(\ref{bound W(t) + w}) and $d+1 \geq  4 > 2^{\Ceta} $, we have $W(t_0) + w \leq W(t) + w \leq 2^d$.
  Hence, by Remark~\ref{sym tree bound}, we can realize $s_0 \land \dotsc \land s_{n-1} \land t_0$ as a binary tree with delay $d$.
  Thus, we will check inductively that $f^*((), (t_1, t_2, \dotsc, t_{m-1})) = g^*((t_1, t_2, \dotsc, t_{m-1}))$
  can be realized with delay $d$.
  Note that requirement (\ref{main req d+1}) and condition (\ref{t_0 case}) imply
  \begin{align*}
   W((t_1, t_2, \dotsc, t_{m-1}))
     & < \Ceta \frac{2^{d} - w}{(d+1)\log_2(d+1)} + \frac{d}{d+1} \Lambda_t - \Ceta \frac{2^{d-1} - w}{d \log_2(d)}\,,
  \end{align*}
  which we claim to be at most $\Ceta \frac{2^{d-1}}{d \log_2(d)} + \frac{d-1}{d} \Lambda_t$.
  This can be shown by
  { 
  \newcommand{\ca}[1]{&\customalign{======}{#1}}
  \begin{align*}
    \ca{} \Ceta \frac{2^{d-1}}{d \log_2(d)}  +  \frac{d-1}{d} \Lambda_t  -  \Ceta \frac{2^{d} - w}{(d+1)\log_2(d+1)}
        -  \frac{d}{d+1} \Lambda_t + \Ceta \frac{2^{d-1} - w}{d \log_2(d)} \\
    \ca{\stackrel{\text{Lem. } \ref{sub-calc}}{\geq}} \Ceta \frac{2^{d-1}\log_2(d+1) - (2^d - w) \log_2(d)}{d \log_2(d) \log_2(d+1)} + \Ceta \frac{2^{d-1} - w}{d \log_2(d)}\\
    \ca{=} \Ceta \frac{(2^{d} - w)(\log_2(d+1) - \log_2(d))}{d \log_2(d) \log_2(d+1)} \\
    \ca{\stackrel{w < 2^d}{\geq}}  0\,.
  \end{align*}
  } 
  Thus, realization~(\ref{case 1 realization}) yields a delay of $d+1$ for $f(s, t)$,
  which proves the lemma for the case that $W(t_0) > \Ceta \frac{2^{d-1} - w}{d \log_2(d)}$.

  \textbf{Case 2:} Assume that $W(t_0) \leq \Ceta \frac{2^{d-1} - w}{d \log_2(d)}$.

  Therefore, we can consider a maximum odd-length prefix $t' = (t_0, t_1, \dotsc, t_{2k})$ of $t$
  with $0 \leq k \leq \frac{m-1}{2}$ such that
  \begin{equation}\label{prefix}
   W(t') \leq \Ceta \frac{2^{d-1} - w}{d \log_2(d)}\,.
  \end{equation}
  We define $t'' := (t_{2k+1}, \dotsc, t_{m-1})$.

  If $t''$ is empty, there is nothing to show since, by induction hypothesis,
  we can construct $f(s, t) = f(s, t')$ with a delay of $d < d+1$ due to $w < 2^{d-1}$.
  Otherwise, we will realize $f(s, t)$ with delay $d+1$ using the alternating split (\ref{alternating split dual})
  for some prefix $t^*$ of $t$ to be determined, i.e.,
  \begin{equation}\label{realization}
   f(s, t) = f(s, t^*) \land f^*(\widehat{t^*}, t^{**}),
  \end{equation}
  where $t^* = (t_0, t_1, \dotsc, t_{2l})$
  for some $0 \leq l < \frac{m-1}{2}$ and $t^{**} := (t_{2l+1}, \dotsc, t_{m-1})$.
  Our main argument, which is presented in Case 2 (ii),
  requires that $\{t_{2k+1}, t_{2k+2}\} \cap \{t_{m-2}, t_{m-1}\} = \emptyset$,
  i.e., that $t''$ has at least $4$ elements.
  Thus, in Case 2 (i), we treat the $t''$ with at most $2$ elements,
  and at the beginning of Case 2 (ii) those with exactly $3$ elements.

  \textbf{Case 2 (i):} Assume that $t''$ consists of at most $2$ elements.

  We set $t^* := t'$, thus $t^{**} = t''$.
  By induction hypothesis and due to $w < 2^{d-1}$,
  inequality~(\ref{prefix}) allows realizing $f(s, t')$ with delay $d$.
  Since $t''$ has at most $2$ elements,
  by Remark \ref{sym tree bound}, we can realize $f^*(\widehat{t'}, t'')$
  as a binary tree with delay $d$ since
  $W(\widehat{t'}) + W(t'') \leq W(t)$, which is at most $2^{d}$
  due to inequality~(\ref{bound W(t) + w}) and $d+1 \geq 4 > 2^{\Ceta}$.

  \textbf{Case 2 (ii):} Assume that $t''$ contains at least $3$ elements.

  Set $\tilde t := (t_0, \dotsc, t_{2k+2})$.
  We need to find an appropriate prefix $t^*$ of $t$ for realization~(\ref{realization})
  such that both $f(s, t^*)$ and $f^*(\widehat{t^*}, t^{**})$
  can be realized with delay $d$ by induction hypothesis.
  We choose $t^*$ depending on the weight of $\tilde t$:
  \begin{enumerate}
  \begin{samepage}
  \renewcommand{\theenumi}{\alph{enumi}}
   \item If $W\left(\tilde t\right) \leq \Ceta \frac{2^{d-1} - w}{d\log_2(d)} + \frac{d-1}{d} \Lambda_{\tilde t}$,
         we set $t^* := \tilde t$.
   \item If $W\left(\tilde t\right) > \Ceta \frac{2^{d-1} - w}{d\log_2(d)} + \frac{d-1}{d} \Lambda_{\tilde t}$,
         we set $t^* := t'$.
         Note that in this case, we in particular have
         \[W(t^*) = W(t') = W\left(\tilde t\right) - \Lambda_{\tilde t} > \Ceta \frac{2^{d-1} - w}{d \log_2(d)} + \frac{d-1}{d} \Lambda_{\tilde t} - \Lambda_{\tilde t}
                = \Ceta \frac{2^{d-1} - w}{d \log_2(d)} - \frac{1}{d}\Lambda_{\tilde t}\,.\]
  \end{samepage}
  \end{enumerate}
  Figure \ref{Lambda figure} visualizes the case distinction.
  In either case, the weight of $t^*$ will be of the form
  \begin{equation}
   \hfill W(t^*) = \Ceta \frac{2^{d-1} - w}{d \log_2(d)} + \delta \hfill \text{ with } \hfill-\frac{1}{d} \Lambda_{\tilde t} \leq \delta \leq \frac{d-1}{d}\Lambda_{\tilde t}\,. \hfill \label{delta formula}
  \end{equation}
  The upper bound on $\delta$ allows us to realize $f(s, t^*)$ with delay $d$ by induction hypothesis since $w < 2^{d-1}$.
  It remains to show that $f^*(\widehat{t^*}, t^{**})$ can be realized with delay $d$.

  \input{lambda_figure.tex}

  The case that $t''$ contains exactly $3$ elements still needs to be treated separately.
  Here, case (a) is easy since we have $t^{**} = (t_{m-1})$,
  hence $f^*(\widehat{t^*}, t^{**})$ is a binary tree
  which can be realized with delay $d$ by Huffman coding since $W(t) \leq 2^d$ due to inequality (\ref{bound W(t) + w}).
  In Case (b), we show that the realization
  \[f^*(\widehat{t'}, (t_{2k+1}, t_{2k+2}, t_{2k+3})) = (\widehat{t'} \lor t_{2k+1}) \lor (t_{2k+2} \land t_{2k+1})\]
  yields delay $d$:
  The binary tree $\widehat{t'} \lor t_{2k+1}$ can be realized with delay $d-1$ using Remark~\ref{sym tree bound} since
  \[W(\widehat{t'}) + W(t_{2k+1}) \overset{(\ref{main req - lambda})}{\leq} \Ceta \frac{2^d - w}{(d+1) \log_2(d+1)}
                                  \overset{\Ceta < 2, w \geq 0}{<} \frac{2^{d+1}}{(d+1) \log_2(d+1)}
                                  \overset{d \geq 2}{\leq} 2^{d-1}\,.\]
   It remains to show $\Lambda_t \leq 2^{d-1}$,
   so assume the contrary.
   W.l.o.g., since $f(s, t)$ is logically symmetric in $t_{2k+2}$ and $t_{2k+3}$,
   we may assume $W(t_{2k+3}) = \max\{W(t_{2k+3}), W(t_{2k+2})\}$.
   Due to inequality~(\ref{bound W(t) + w}) and $m \geq 4$, we have $2^{d-1} < \Lambda_t < 2^d$.
   This implies $W(t_{2k+3}) =2^{d-1}$ and $W(t_{2k+2}) \leq 2^{d-2}$.
   By combining
   { 
  \newcommand{\ca}[1]{&\customalign{=====}{#1}}
   \begin{align*}
    W(t) - \Lambda_t \overset{(\ref{main req d+1})}{\leq} \Ceta \frac{2^{d} - w}{(d+1) \log_2(d+1)} - \frac{\Lambda_t}{d+1}
      \ca{=} \Ceta \frac{2^{d} - w}{(d+1) \log_2(d+1)} - \frac{W(t_{2k+2}) + W(t_{2k+3})}{d+1}
   \shortintertext{and}
             W(t) - \Lambda_t = W(\tilde t) - W(t_{2k+2})
      \ca{\overset{\text{case (b)}}{>}} \Ceta \frac{2^{d-1} - w}{d \log_2 d} + \frac{d-1}{d} \Lambda_{\tilde t} - W(t_{2k+2}) \\
      \ca{\overset{(\ref{W >= 1})}{\geq}} \Ceta \frac{2^{d-1} - w}{d \log_2 d} + \frac{d-1}{d} - \frac{W(t_{2k+2})}{d} \,,
   \end{align*}
   }
   we obtain
   { 
  \newcommand{\ca}[1]{&\customalign{=========}{#1}}
   \begin{align*}
    0 \ca{<} \Ceta \frac{2^{d} - w}{(d+1) \log_2(d+1)} - \Ceta \frac{2^{d-1} - w}{d \log_2 d} - \frac{W(t_{2k+2}) + W(t_{2k+3})}{d+1} - \frac{d-1}{d} + \frac{W(t_{2k+2})}{d} \\
      \ca{\overset{w < 2^{d-1}}{<}} \Ceta \frac{2^{d-1}}{(d+1) \log_2(d+1)} + \frac{W(t_{2k+2}) - d W(t_{2k+3}) - d^2 + 1}{d(d+1)} \\
      \ca{\overset{\substack{W(t_{2k+2}) \leq 2^{d-2}, \\ W(t_{2k+3}) = 2^{d-1}}}{\leq}} \Ceta \frac{2^{d-1}}{(d+1) \log_2 (d+1)} + \frac{2^{d-2} - d 2^{d-1} - d^2 + 1}{d(d+1)}  \\
      \ca{\overset{\Ceta < 2}{<}} \frac{d 2^{d} + \log_2 (d+1) \left(2^{d-2} - d 2^{d-1} - d^2 + 1\right) }{d (d+1) \log_2(d+1)} \\
      \ca{<} 0\,,
   \end{align*}
   }
   where the last step can be verfied by hand for $d = 3$,
   and for $d \geq 4$ is implied by
   \[d 2^d + \log_2(d+1) \left(2^{d-2} - d 2^{d-1}\right) = 2^{d-2} \left(4d + \log_2 (d+1) \left( 1 - 2d \right) \right)
      \overset{d \geq 4}{\leq} 2^{d-2} \left(4 d - 7 \log_2 (5) \right) \overset{d \geq 4}{<} 0\,.\]
   This is a contradiction, concluding the case that $t''$ has exactly $3$ elements.

  Now we may assume that $t''$ contains at least $4$ elements.
  In particular, the first two elements $t_{2k+1}$ and $t_{2k+2}$ of $t''$ and the last two elements $t_{m-2}$ and $t_{m-1}$ of $t$ are disjoint sets.
  Note that since $t^*$ does not contain any of the last two elements of $t$, we have
  \[W(\widehat{t^*}) \stackrel{(\ref{main req - lambda})}{\leq} \Ceta \frac{2^{d} - w}{(d+1) \log_2(d+1)}
       \stackrel{w \geq 0}{\leq} \Ceta \frac{2^d}{(d+1) \log_2(d+1)}
      \stackrel{\Ceta < 2}{<} 2^{d-1}\]
  for $d \geq 2$
  and thus, by induction hypothesis, it suffices to prove that
  \begin{equation}\label{to show}
   W(t^{**}) \leq \Ceta \frac{2^{d-1} - W(\widehat{t^*})}{d \log_2(d)} + \frac{d-1}{d} \Lambda_{t^{**}}\,.
  \end{equation}
  Due to requirement (\ref{main req d+1}), we have
  \[W(t^{**}) = W(t) - W(t^*) \leq \Ceta \frac{2^{d}-w}{(d+1) \log_2(d+1)} + \frac{d}{d+1} \Lambda_t - W(t^*)\,.\]
  Since $W(\widehat{t^*}) \leq W(t^*)$ and $\Lambda_{t^{**}} = \Lambda_t$,
  inequality~(\ref{to show}) is thus implied if we prove the following claim.

  \begin{claim}
   We have $\Ceta \frac{2^{d-1} - W(t^*)}{d \log_2(d)} + \frac{d-1}{d} \Lambda_{t}
       - \Ceta \frac{2^{d}-w}{(d+1) \log_2(d+1)} - \frac{d}{d+1} \Lambda_t + W(t^*) \geq 0\,.$
  \end{claim}
  \begin{proof_of_claim}
  We first only bound the summands depending on $W(t^*)$ or $\Lambda_t$.
 { 
  \allowdisplaybreaks
  \newcommand{\ca}[1]{&\customalign{=====}{#1}}
  \begin{align}
  \ca{} - \Ceta \frac{W(t^*)}{d \log_2(d)} + \frac{d-1}{d} \Lambda_t - \frac{d}{d+1} \Lambda_t + W(t^*) \nonumber\\
  \ca{=} W(t^*) \left(1 - \frac{\Ceta}{d \log_2(d)}\right) - \frac{1}{d (d+1)} \Lambda_t \nonumber\\
  \ca{\stackrel{(\ref{delta formula})}{\geq}} \left(\Ceta \frac{2^{d-1} - w}{d \log_2(d)}
            - \frac{1}{d}\Lambda_{\tilde t}\right) \frac{d \log_2(d) - \Ceta}{d \log_2(d)} - \frac{1}{d (d+1)} \Lambda_t \nonumber\\
  \ca{=} \Ceta \frac{(2^{d-1} - w)(d \log_2(d) - \Ceta )}{d^2 \log^2_2(d)} - \frac{d \log_2(d) - \Ceta}{d^2 \log_2(d)} \Lambda_{\tilde t}
            - \frac{1}{d (d+1)} \Lambda_t \nonumber\\
  \ca{=} \Ceta \frac{(2^{d-1} - w)(d \log_2(d) - \Ceta)}{d^2 \log^2_2(d)} - \frac{1}{d} \left( \Lambda_{\tilde t} + \frac{\Lambda_t}{d+1} \right) + \frac{\Ceta \Lambda_{\tilde t}}{d^2 \log_2(d)} \nonumber\\
  \ca{\stackrel{(\ref{main req - lambda})}{\geq}} \Ceta \frac{(2^{d-1} - w)(d \log_2(d) - \Ceta)}{d^2 \log^2_2(d)}
      - \frac{1}{d} \left( \Ceta \frac{2^d - w}{(d+1)\log_2(d+1)} - W(t') \right) + \frac{\Ceta \Lambda_{\tilde t}}{d^2 \log_2(d)} \nonumber \\
  \ca{\stackrel{\text{Rem. }\ref{W >= 1}}{\geq}} \Ceta \frac{(2^{d-1} - w)(d \log_2(d) - \Ceta)}{d^2 \log^2_2(d)}
      - \Ceta \frac{2^d - w}{d(d+1)\log_2(d+1)} + \frac{1}{d} + \frac{2 \Ceta}{d^2 \log_2(d)} \label{Lambda delta calc}
  \end{align}
  } 
  Note that in the last two steps, we used that $t_{2k+1}, t_{2k+2}, t_{m-2}, t_{m-1}$ are four different inputs
  which are not contained in $t'$ and that $t'$ is not empty.
  Based on inequality~(\ref{Lambda delta calc}), the left-hand side of the claim can be bounded from below by
  { 
  \newcommand{\ca}[1]{&\customalign{===}{#1}}
  \begin{align*}
   \ca{} \Ceta \frac{2^{d-1} - W(t^*)}{d \log_2(d)} + \frac{d-1}{d} \Lambda_{t}
          - \Ceta \frac{2^{d}-w}{(d+1) \log_2(d+1)} - \frac{d}{d+1} \Lambda_t + W(t^*) \\
   \ca{\stackrel{(\ref{Lambda delta calc})}{\geq}} \Ceta \Bigg( \frac{(2^{d-1} - w)(d \log_2(d) - \Ceta)}{d^2 \log^2_2(d)}
          - \frac{2^d - w}{d(d+1)\log_2(d+1)}
          + \frac{1}{d} + \frac{2 \Ceta}{d^2 \log_2(d)}
          + \frac{2^{d-1}}{d \log_2(d)} \\
   \ca{} \phantom{\Ceta \Bigg(}
          - \frac{2^{d}-w}{(d+1) \log_2(d+1)} \Bigg) \\
   \ca{=} \Ceta \left( \frac{2^{d-1} - w + 2^{d-1}}{d \log_2(d)}
                      - \Ceta \frac{2^{d-1} - w}{d^2 \log_2^2(d)}
                      - \frac{(2^{d}-w) (d+1)}{d (d+1)\log_2(d+1)}
                      + \frac{1}{d} + \frac{2 \Ceta}{d^2 \log_2(d)}\right) \\
   \ca{=} \Ceta \left( \frac{2^{d} - w}{d \log_2(d)}
                       - \Ceta \frac{2^{d-1} - w}{d^2 \log_2^2(d)}
                       - \frac{2^{d}-w}{d \log_2(d+1)}
                       + \frac{1}{d} + \frac{2 \Ceta}{d^2 \log_2(d)}\right) \\
   \ca{=}  \frac{\Ceta}{d^2 \log_2^2(d) \log_2(d+1)}
    \Bigg(\log_2(d+1) \Big((2^d - w) d \log_2(d)
                - \Ceta(2^{d-1} - w)\Big) \\
   \ca{} \phantom{\frac{\Ceta}{d^2 \log_2^2(d) \log_2(d+1)} +}
                - (2^d-w) d \log_2^2(d)
                + \left(2 + \frac{1}{\Ceta} d \log_2(d)\right) \log_2(d) \log_2(d+1)\Bigg) \,,
  \end{align*}
  } 
  which is required to be non-negative.
  After multiplying with the denominator and dividing by $\Ceta$,
  we apply the bound $\log_2(d+1) \geq \log_2(d) + \frac{1}{d}$ stated in Remark \ref{log remark},
  and thus can prove the claim if we show that
  { 
  \newcommand{\ca}[1]{&\customalign{======}{#1}}
  \begin{align*}
   \ca{} \log_2(d+1) \Big((2^d - w) d \log_2(d) - \Ceta(2^{d-1} - w)\Big) \\
   \ca{}           - (2^d-w) d \log_2^2(d)
                   + \left(2 + \frac{1}{\Ceta} d \log_2(d)\right) \log_2(d) \log_2(d+1)\\
   \ca{\stackrel{\Ceta < 2, (\ref{log(d+1) greater})}{\geq}} (2^d - w) \log_2(d)
                                                             - \left(\log_2(d) + \frac{1}{d}\right) \Ceta(2^{d-1} - w)
                                                             + \left(2 + \frac{1}{\Ceta} d \log_2(d)\right) \log_2(d) \log_2(d+1)\\
   \ca{\stackrel{\Ceta \geq 1, w \geq 0}{\geq}} 2^d \log_2(d)
                                                - \left(\log_2(d) + \frac{1}{d}\right) \Ceta 2^{d-1}
                                                + \left(2 + \frac{1}{\Ceta} d \log_2(d)\right) \log_2(d) \log_2(d+1)
   \end{align*}
   } 
  is at least $0$.
  Note that for $d \geq 7$, this is already implied by
  \[2^d \log_2(d) - \left(\log_2(d) + \frac{1}{d}\right) \Ceta 2^{d-1}
   \stackrel{\Ceta = 1.9}{=} 2^{d-1} \left (0.1 \log_2 (d) - \frac{1.9}{d} \right)
   \stackrel{d \geq 7}{\geq} 0\,.\]
   For $3 \leq d \leq 6$, we have
   { 
  \newcommand{\ca}[1]{&\customalign{====}{#1}}
   \begin{align*}
    \ca{} 2^d \log_2(d) - \left(\log_2(d) + \frac{1}{d}\right) \Ceta 2^{d-1}
                        + \left(2 + \frac{1}{\Ceta} d \log_2(d)\right) \log_2(d) \log_2(d+1)\\
    \ca{\stackrel{\Ceta = 1.9}{=}} \log_2(d) \left(0.1 \cdot 2^{d-1} + \left(2 + \frac{1}{1.9} d \log_2(d)\right) \log_2(d+1)\right) - \frac{1.9}{d} 2^{d-1} \\
    \ca{\stackrel{d \geq 3}{\geq}} \log_2(3) \left(0.1 \cdot 2^{2} + \left(2 + \frac{1}{1.9} 3 \log_2(3)\right) 2\right)
                                   - \frac{1.9 }{d} 2^{d-1} \\
    \ca{>} 14 - \frac{1.9 }{d} 2^{d-1} \\
    \ca{\stackrel{d \leq 6}{\geq}} 14 - \frac{1.9 }{6} 32\\
    \ca{>} 0\,.
   \end{align*}
  } 
  This proves the claim.
   Note that this is the only place where we used the definition $\Ceta = 1.9$.
  \end{proof_of_claim}

  Thus, by induction hypothesis, we can find a realization with delay $d$ for $f(\widehat{t^*}, t^{**})$.
  Split (\ref{realization}) hence also provides a realization with delay $d+1$ for $f(s, t)$ if $t''$ contains at least $4$ elements.
  This concludes the proof.
  \end{proof}
\end{lemma}

\begin{proof}[Proof of Theorem \ref{main theorem d+1}]
 Lemma \ref{small w} proves the theorem in the case that $0 \leq w < 2^{d-1}$,
 while Lemma \ref{large w} proves it for the remaining case that $2^{d-1} \leq w < 2^d$.
\end{proof}

Finally, we can prove Theorem \ref{main theorem}.

\begin{proof}[Proof of Theorem \ref{main theorem}]
 We prove the theorem by induction on $d$.
 For $d \leq 3$, Lemma \ref{base case} provides a realization of $f(s, t)$ with delay $d$.
 Now we can assume that the theorem holds for some $d \geq 3$,
 and prove the inductive step via Theorem \ref{main theorem d+1}.
\end{proof}

%% file: lambda_figure.tex
\begin{figure}
   \centering
   \providecommand{\subfigw}{0.7\textwidth}
   \begin{minipage}{\subfigw}
   \resizebox{\textwidth}{!}{%
   \begin{tikzpicture}
   \tikzstyle{every node}=[outer sep=0pt]

   \node (poss2) at (-2.5, 3.2){\phantom{$t_0$}};
   \node (post0) at (-1.5, 3.2){\phantom{$t_0$}};
   \node (post4) at (2.5, 3.2){\phantom{$t_0$}};
   \node (post5) at (3.5, 3.2){\phantom{$t_0$}};
   \node (post6) at (4.5, 3.2){\phantom{$t_0$}};
   \node (post7) at (5.5, 3.2){\phantom{$t_0$}};

   \draw [fill=green, opacity=0.5] ($(poss2.north)!0.5!(post0.north)$) rectangle ($(post4.south)!0.5!(post5.south)$);
   \draw [fill=yellow, opacity=0.5] ($(post4.south)!0.5!(post5.south)$) rectangle ($(post6.north)!0.5!(post7.north)$);

   \node (s0) at (-4.5, 3.2){$s_0$};
   \node (s1) at (-3.5, 3.2){$s_1$};
   \node (s2) at (poss2){$s_2$};
   \node (t0) at (-1.5, 3.2){$t_0$};
   \node (t1) at (-0.5, 3.2){$t_1$};
   \node (t2) at (0.5, 3.2){$t_2$};
   \node (t3) at (1.5, 3.2){$t_3$};
   \node (t4) at (post4){$t_4$};
   \node (t5) at (post5){$t_5$};
   \node (t6) at (post6){$t_6$};
   \node (t7) at (post7){$t_7$};
   \node (t8) at (6.5, 3.2){$t_8$};
   \node (t9) at (7.5, 3.2){$t_9$};

   \draw[decorate,decoration={brace, amplitude=5pt, mirror, raise=10pt}] (t5.west) -- (t6.east) node [midway, below, sloped, yshift=-15pt] {$\Lambda_{\tilde t}$};
   \draw[decorate,decoration={brace, amplitude=5pt, mirror, raise=10pt}] (t8.west) -- (t9.east) node [midway, below, sloped, yshift=-15pt] {$\Lambda_{t}$};
   \draw[decorate,decoration={brace, amplitude=10pt, raise=10pt}] (t0.west) -- (t4.east) node [midway, above, sloped, yshift=20pt] {$t'$};
   \draw[decorate,decoration={brace, amplitude=10pt, raise=30pt}] (t0.west) -- (t6.east) node [midway, above, sloped, yshift=40pt] {$t^* := \tilde t$};
   \draw[decorate,decoration={brace, amplitude=10pt, raise=30pt}] (t7.west) -- (t9.east) node [midway, above, sloped, yshift=40pt] {$t^{**}$};

   \end{tikzpicture}}
   \caption*{(a) \, In the case that $W\left(\tilde t\right) \leq \frac{2^{d-1} - w}{d \log_2(d)}$, we set $t^* := \tilde t$.}
   \end{minipage}
   \begin{minipage}{\subfigw}
   \resizebox{\textwidth}{!}{%
   \begin{tikzpicture}

   \node (poss2) at (-2.5, 3.2){\phantom{$t_0$}};
   \node (post0) at (-1.5, 3.2){\phantom{$t_0$}};
   \node (post4) at (2.5, 3.2){\phantom{$t_0$}};
   \node (post5) at (3.5, 3.2){\phantom{$t_0$}};
   \node (post6) at (4.5, 3.2){\phantom{$t_0$}};
   \node (post7) at (5.5, 3.2){\phantom{$t_0$}};

   \draw [fill=green, opacity=0.5] ($(poss2.north)!0.5!(post0.north)$) rectangle ($(post4.south)!0.5!(post5.south)$);
   \draw [fill=red, opacity=0.5] ($(post4.south)!0.5!(post5.south)$) rectangle ($(post6.north)!0.5!(post7.north)$);

   \node (s0) at (-4.5, 3.2){$s_0$};
   \node (s1) at (-3.5, 3.2){$s_1$};
   \node (s2) at (poss2){$s_2$};
   \node (t0) at (-1.5, 3.2){$t_0$};
   \node (t1) at (-0.5, 3.2){$t_1$};
   \node (t2) at (0.5, 3.2){$t_2$};
   \node (t3) at (1.5, 3.2){$t_3$};
   \node (t4) at (post4){$t_4$};
   \node (t5) at (post5){$t_5$};
   \node (t6) at (post6){$t_6$};
   \node (t7) at (post7){$t_7$};
   \node (t8) at (6.5, 3.2){$t_8$};
   \node (t9) at (7.5, 3.2){$t_9$};

   \draw[decorate,decoration={brace, amplitude=5pt, mirror, raise=10pt}] (t5.west) -- (t6.east) node [midway, below, sloped, yshift=-15pt] {$\Lambda_{\tilde t}$};
   \draw[decorate,decoration={brace, amplitude=5pt, mirror, raise=10pt}] (t8.west) -- (t9.east) node [midway, below, sloped, yshift=-15pt] {$\Lambda_{t}$};
   \draw[decorate,decoration={brace, amplitude=10pt, raise=10pt}] (t0.west) -- (t4.east) node [midway, above, sloped, yshift=20pt] {$t^* := t'$};
   \draw[decorate,decoration={brace, amplitude=10pt, raise=10pt}] (t5.west) -- (t9.east) node [midway, above, sloped, yshift=20pt] {$t^{**}$};
   \draw[decorate,decoration={brace, amplitude=10pt, raise=30pt}] (t0.west) -- (t6.east) node [midway, above, sloped, yshift=40pt] {$\tilde t$};

   \end{tikzpicture}}
   \caption*{(b) \, In the case that $W\left(\tilde t\right) > \frac{2^{d-1} - w}{d \log_2(d)}$, we set $t^* := t'$.}
   \end{minipage}
   \caption{Illustration of the choice of $t^*$.}
   \label{Lambda figure}
  \end{figure}

%% file: bounding_delay.tex
\section{Constructing Fast Circuits}\label{sec::bound_delay}

Based on Theorem~\ref{bound on v(d, w)}, we could now show that there is a circuit realizing the {\aop }
$t_0 \land (t_1 \lor (t_2 \land ( \dots t_{m-1}) \dots )$ with delay at most
$\log_2 W + \log_2 \log_2 W + \log_2 \log_2 \log_2 W + 5$.
Instead, we will prove a stronger result:
By modifying the instance, we can diminish the dependency on $W$.
The modification is based on the observation that we
can round up small arrival times to the same value
without losing too much for the maximum delay. Moreover, shifting
all arrival times by some number does not change the problem.
Both modifications allow us to reduce the problem to instances with
a total arrival time weight of at most $2m$.

\begin{theorem}\label{theorem::improved_delay2}
Let $m \in \N$ with $m \geq 3$, Boolean variables $t_0,\dots,t_{m-1}$
and arrival times $a : \{t_0,\dots,t_{m-1}\} \to \N$ be given,
and define $W:= \sum_{i=0}^{m-1} 2^{a(t_i)}$.
There is circuit realizing the \aop~
$t_0 \land (t_1 \lor (t_2 \land ( \dots t_{m-1}) \dots )$
with delay at most
\[
   \log_2 W  + \log_2 \log_2 m + \log_2 \log_2 \log_2 m + 4.3\,.
\]
\proof
We compute new arrival times $\tilde a: \{t_0,\dots,t_{m-1}\} \to \N$ by
setting
\[
   \tilde a(t_i) := \max\{0, a(t_i) - \lceil (\log_2 W - \log_2 m) \rceil \}
\]
for all $i \in \{0,\dots,m-1\}$. We define $\widetilde W := \sum_{i=0}^{m-1} 2^{\tilde a(t_i)}$
and partition the input indices into $I_1 := \{i \in \{0,\dots,m-1\} \mid \tilde a(t_i) = 0\}$ and
$I_2 := \{0,\dots,m-1\} \setminus I_1$.
Then, we have
\begin{eqnarray*}
   \widetilde W
            = \sum_{i \in I_1} 2^{0} +  2^{- \lceil (\log_2 W - \log_2 m) \rceil}\sum_{i \in I_2} 2^{a(t_i)}
            \leq m + \frac{2^{\log_2 m}}{2^{\log_2 W}} W
            = 2 m\,.
\end{eqnarray*}
Define $\tilde d := \lfloor \log_2 m + \log_2 \log_2 m + \log_2 \log_2 \log_2 m + 3.3 \rfloor$.

\begin{claim}
There is a circuit $C$ realizing the {\aop } $t_0 \land (t_1 \lor (t_2 \land ( \dots t_{m-1}) \dots )$
with arrival times $\tilde a$ with delay at most $\tilde d$.
\end{claim}

\begin{proof_of_claim}
Let $M:= 500$.
If $m < M$, we have
$
  1.441 \log_2 \widetilde W + 2.674
     \leq 1.441 \log_2 (2 m) + 2.674
     = 1.441 \log_2 m + 4.115  \leq  \log_2 m + \log_2 \log_2 m  + \log_2 \log_2 \log_2 m  + 3.3$.
Since the \aop~optimization method by Held and Spirkl \cite{Spirkl} computes a
circuit with delay at most $\lfloor 1.441 \log_2 \widetilde W + 2.674 \rfloor$, this
proves the claim for $m < M$.

Hence assume $m \geq M$.
For proving the claim, by Theorem~\ref{bound on v(d, w)}, it is sufficient to
show
\begin{equation}\label{eq::bound::2m}
   2 m \leq \Ceta \frac{2^{\tilde d-1}}{\tilde d \log_2 \tilde d}\,.
\end{equation}
Note that the mapping $x \mapsto \frac{2^{x-1}}{x \log_2 x}$ is strictly increasing for
$x \geq 2$. Moreover, we have
$ \tilde d \geq \log_2 m + \log_2 \log_2 m + \log_2 \log_2 \log_2 m + 2.3$.
Since
\begin{equation} \label{eq::logbound33}
  \log_2 m + \log_2 \log_2 m + \log_2 \log_2 \log_2 m + 2.3
   \leq 1.8 \log_2 m
\end{equation}
for $m \geq M$, equation~(\ref{eq::bound::2m}) is hence valid if
\[
   2 m \leq \Ceta  \frac{\frac{1}{2} \cdot m \cdot  \log_2 m \cdot \log_2 \log_2 m \cdot 2^{2.3}}
                {1.8 \cdot \log_2 m \cdot \log_2(1.8 \cdot \log_2 m)}\,.
\]
This is equivalent to
\begin{equation}\label{eq::bound8}
   1.8 \cdot \log_2 (1.8 \cdot \log_2 m) \leq \Ceta 2^{0.3} \cdot \log_2 \log_2 m\,,
\end{equation}
which is true for $m \geq M$ since $\Ceta = 1.9$.
This proves the claim.
\end{proof_of_claim}

Since we have $a(t_i)  \leq  \tilde a(t_i) + \lceil (\log_2 W  -  \log_2 m) \rceil$
for all $i \in \{0,\dots,m-1\}$, the circuit $C$ has, for the initial arrival
times $a: \{t_0,\dots, m-1 \} \to \N$, a delay of at most
{ 
  \newcommand{\ca}[1]{&\customalign{==}{#1}}
\begin{align*}
        \ca{} \lfloor \log_2 m + \log_2 \log_2 m + \log_2 \log_2 \log_2 m + 3.3\rfloor + \lceil (\log_2 W -  \log_2 m ) \rceil \\
   \ca{\leq}  \log_2 W  + \log_2 \log_2 m + \log_2 \log_2 \log_2 m + 4.3\,. \qedhere
\end{align*}
  } 
\end{theorem}

\begin{rem}
In the proof, we apply the algorithm from Held and Spirkl \cite{Spirkl} for
small instances. Without this trick, we would obtain a delay bound of
$\log_2 W + \log_2 \log_2 m + \log_2 \log_2 \log_2 m + 7$.
Moreover, for sufficiently large values of $m$, the delay bound in the previous theorem
can be improved slightly to
$   \log_2 W + \log_2 \log_2 m + \log_2 \log_2 \log_2 m + 4 +  \varepsilon
$
for any constant $\varepsilon > 0$:
Note that the factor $1.8$ in inequality~(\ref{eq::logbound33}) can be
decreased to a value arbitrarily close to $1$ if $m$ is sufficiently large.
Thus, also the factor $\Ceta 2^{0.3}$ in inequality~(\ref{eq::bound8})
becomes arbitrarily close to $1$ for large values of $m$.
This leads to the stated delay bound.
\end{rem}

The following theorem shows that the circuit described in Theorem \ref{theorem::improved_delay2}
does not only exist, but can also be computed efficiently.

\begin{theorem} \label{thm:practice}
 There is an algorithm that computes the circuit in Theorem~\ref{theorem::improved_delay2} for given $m \geq 3$
 in time $\mathcal{O}(m^2 \log_2 m)$.
\begin{proof}
As main subroutine, we will use Algorithm \ref{alg}.

\begin{claim}
Given input variables $s = (s_1,\dots,s_{n-1})$
and $t = (t_0,\dots,t_{m-1})$ with arrival times $a: \{t_0,\dots,t_{m-1}, s_0, \dotsc, s_{n-1}\} \to \N$,
Algorithm~\ref{alg} computes a Boolean circuit realizing $f(s, t)$
with delay at most $d$, where $d$ is the smallest
natural number with ${w:= W(s) < 2^{d-1}}$ and
$W(t) \leq \Ceta \frac{2^{d-1} - w}{d \log_2 (d)} + \frac{d-1}{d} \Lambda_t$.
The number of computation steps of Algorithm~\ref{alg} is
\[{\mathcal{O}(m (m+n) \log_2(m+n) + m\log_2\log_2 (W'))}\,,\]
where $W' = \sum_{i=0}^{m-1} 2^{a(t_i)} +  \sum_{i=0}^{n-1} 2^{a(s_i)}$.
\end{claim}

\begin{proof_of_claim}
We apply the recursive approach described in Algorithm~\ref{alg}
which arises from the proof of Theorem~\ref{main theorem}:
In line~\ref{algo compute d}, we compute the minimum $d \in \N$ such that
$W(t) \leq \Ceta \frac{2^{d-1} - w}{d \log_2 (d)} + \frac{d-1}{d} \Lambda_t$.
We have $d \in \mathcal O(\log_2 (W'))$, so $d$ can be computed by binary search in
$\mathcal O(\log_2 \log_2 (W'))$ steps.
Note that in line~\ref{algo compute d}, we have $w < 2^{d-1}$ since otherwise,
we would obtain a contradiction to $\Lambda_t \leq W(t)$ since
\begin{align*}
   W(t) &\leq \Ceta \frac{2^{d-1} - w}{d \log_2(d)} + \frac{d-1}{d} \Lambda_t < \frac{d-1}{d} \Lambda_t < \Lambda_t\,.
\end{align*}
Thus, Theorem~\ref{main theorem}
provides a circuit realizing $f(s, t)$ with delay $d$.
Lemma~\ref{base case} computes this realization if $d \leq 3$ (see lines \ref{algo small d start} to \ref{algo small d end}).
For $d > 3$, Lemmata~\ref{large w} (see lines \ref{algo large w start} to \ref{algo large w end})
and~\ref{small w} (see lines \ref{algo small w start} to \ref{algo small w end})
construct the circuit recursively.
Hence, the claimed delay bound is fulfilled by Theorem~\ref{main theorem}.

We prove the bound on the number of computation steps of Algorithm \ref{alg}
by counting the number of steps needed for a single call
excluding the recursive calls (i.e., lines
\ref{algo large w realization},
\ref{algo case 1 rec real},
\ref{algo realization f(s, t^*)},
\ref{algo small w recursion})
and bounding the number of recursion steps.

Note that whenever two non-disjoint sequences of inputs are considered as
alternating inputs in the algorithm, one of the sequences must be a subset
of the other. Therefore, the number of alternating input sequences considered
by the algorithm can be bounded by $m$.
Moreover, in each of the recursive calls in lines
\ref{algo case 1 rec real},
\ref{algo realization f(s, t^*)},
\ref{algo small w recursion},
the number of alternating inputs decreases by $1$,
and in the only other recursive call in line \ref{algo large w realization},
the number remains the same,
but in this case, we recursively compute $f((), t)$,
thus the number of alternating inputs will decrease in the next recursive call.
Thus, the number of recursive calls in the algorithm is bounded by $2m$.

Note that in each call of Algorithm~\ref{alg}, we compute at most one symmetric binary tree.
Since each symmetric tree we construct has at most $m+n$ inputs,
due to Remark~\ref{sym tree bound}, this takes $\mathcal{O}((m+n) \log_2 (m+n))$ steps per tree.


If we precompute the weight for each consecutive subset of $t$,
computing the prefix $t'$ in line \ref{compute t'} (or
finding out that no such prefix exists) requires $\mathcal O (\log_2 m)$ steps using binary search.

Apart from this, there are only constantly many steps in each recursive call.

Hence, the number of steps needed for each recursive call of Algorithm \ref{alg},
excluding lines
\ref{algo large w realization},
\ref{algo case 1 rec real},
\ref{algo realization f(s, t^*)}, and
\ref{algo small w recursion},
is at most $\mathcal O((m+n)\log_2(m+n) + \log_2\log_2(W'))$.
Since there are at most $2m$ recursive calls,
we have at most $\mathcal O(m(m+n) \log_2 (m+n) + m \log_2\log_2(W'))$ steps in total,
which finishes the proof of the claim.
\end{proof_of_claim}

Now we can prove the theorem.
We follow the proof of Theorem~\ref{theorem::improved_delay2},
also using its notation.
For $m < M$, we construct the circuit described in \cite{Spirkl}
such that nothing more is to show due to the properties collected in Table~\ref{non-uniform bounds}.
For $m \geq M$,
we compute the modified instance with arrival times $\tilde a$ and weight $\widetilde W$
in linear time. Then, we call Algorithm~\ref{alg} with the modified arrival times $\tilde a$.
Since $\widetilde W \in \mathcal O(m)$,
the sizes of all numbers occurring in the algorithm are polynomial in $m$.
Applying the claim with $s=()$, hence $n = 0$, and $W' = \widetilde W$, we obtain a running time of
$\mathcal{O}(m^2 \log_2(m))$.
\end{proof}
\end{theorem}

\pagebreak
\input{algorithm.tex}
\pagebreak

Our main objective when designing good circuits for \aop s is delay.
Still, there are other metrics to be regarded during circuit construction
such as the size, i.e., the total number of gates used in the circuit,
and maximum fanout, i.e., the maximum number of successors of any input or gate.

\begin{theorem} \label{thm:size_fanout}
 The circuit computed in Theorem~\ref{thm:practice}
 has size at most \[ m \left(\log_2 m + \log_2 \log_2 m + \log_2 \log_2 \log_2 m + 3.3\right) \] and
 maximum fanout at most \[\log_2 m + \log_2 \log_2 m + \log_2 \log_2 \log_2 m + 3.3\,.\]
\end{theorem}
\begin{proof}
In order to prove the fanout bound, we show the following claim.

\begin{claim}
In the circuit computed by Algorithm \ref{alg},
each gate has fanout exactly $1$,
each input in $s$ has fanout exactly $1$ and each input in $t$ has fanout at most $d$.
\begin{proof_of_claim}
Note that each gate constructed has fanout $1$.
We prove the bound on the maximum fanout of the inputs by induction on $d$.

Note that in the realizations computed by Lemma~\ref{base case}
which is used in lines \ref{algo small d start} to \ref{algo small d end},
each input has fanout $1$.
In the realizations provided in lines \ref{algo large w end} and \ref{algo case 1},
inputs of $s$ only occur in exactly one binary tree and thus have a fanout of $1$.
Since we can inductively assume that $f((), t)$ and $f((), (t_1, \dotsc, t_{m-1}))$
fulfill the claimed fanout bound, respectively,
each input in $t$ has fanout at most $d-1 < d$.

In the realization in line \ref{algo small w end},
inductively, the circuit for $f(s, t^*)$ computed in line \ref{algo small w recursion}
has fanout at most $1$ for inputs in $s$
and fanout at most $d-1$ for inputs in $t^*$.
Since inputs of $s$ do not occur in $f^*(\widehat{t^*}, t^{**})$,
it remains to show that inputs of $t$ have fanout at most $d$ in the realization for $f(s, t)$.
If $t''$ has at most $3$ inputs, lines \ref{algo small w <= 2} and \ref{algo small w = 3} show that each input of $t$
has fanout at most $1$ in the realization of $f^*(\widehat{t^*}, t^{**})$, which proves the claimed fanout bounds.
Otherwise, we realize $f^*(\widehat{t^*}, t^{**})$ recursively in line \ref{algo small w recursion}
and inductively can assume that the inputs of $\widehat{t^*}$ have fanout at most $1$
and the inputs of $t^{**}$ fanout at most $d-1$ in this realization.
Together with the recursive fanout bounds for the realization of $f(s, t^*)$,
this shows the claimed fanout bounds for the circuit constructed for $f(s, t)$.
This proves the claim.
\end{proof_of_claim}
\end{claim}

For the bound on the size, recall that the circuit we construct is a formula.
Hence, the size is exactly $\sum_{i} \text{fanout}(t_i) - 1$.
By the proven fanout bound, this is at most
\[m \left(\log_2 m + \log_2 \log_2 m + \log_2 \log_2 \log_2 m + 3.3\right)\,. \qedhere\]
\end{proof}

\begin{rem}
 For given $m \geq 3$, with the additional use of buffers,
 the circuit constructed in Theorem \ref{thm:practice}
 can be transformed into a logically equivalent circuit
 with maximum fanout $2$, but
 with delay at most
 \[\log_2 W + 2 \log_2 \log_2 m + \log_2 \log_2 \log_2 m + 6\]
 and size at most $\mathcal O(m \log_2 m \log_2 \log_2 m)$.

 To see this, first note that the circuit constructed in \cite{Spirkl}
 which we use for small instances already has a maximum fanout of $2$.
 Secondly, note that the circuit $C$ constructed in Theorem~\ref{thm:practice}
 has fanout larger than $1$ only at the inputs.
 Write
 \[f := \log_2 m + \log_2 \log_2 m + \log_2 \log_2 \log_2 m + 3.3\]
 for the maximum possible fanout of $C$.
 For each input $t_i$, we can replace the outgoing edges of $t_i$
 by a delay-optimum buffer tree with maximum fanout $2$ for each buffer (compare Remark~\ref{sym tree bound}).
 This increases the size by at most $m(f - 1)$ and,
 since we can assume that $m \geq 500$, the delay by at most $\log_2 f \leq \log_2 \log_2 m + 1$.
 This yields the stated properties of the transformed circuit.
\end{rem}

%% file: algorithm.tex
\LinesNumbered
 \begin{algorithm}[H]
   \DontPrintSemicolon
    \KwIn{Inputs $s = (s_0, \dotsc, s_{n-1})$ and $t = (t_0, \dotsc, t_{m-1})$, arrival times $a(s_i), a(t_i) \in \N$.}
    \KwOut{Circuit $C$ computing $f(s, t)$.}
    \BlankLine
    Set $w := W(s)$.
    Choose $d \in \N$ minimum such that $W(t) \leq \Ceta \frac{2^{d-1} - w}{d \log_2 (d)} + \frac{d-1}{d} \Lambda_t$. \; \label{algo compute d}
    \If{$d \leq 3$}
    {
      \If{$m \leq 2$}
      {
         Obtain a circuit for $f(s, t)$ as a binary tree $s_0 \land \dotsc \land s_{n-1} \land t_0 \land \dotsc \land t_{m-1}$.\; \label{algo small d start}
      }
      \Else
      {
         Obtain a circuit for $f(s, t)$ via $f(s, t) = t_0 \land (t_1 \lor t_2)$.\; \label{algo small d end}
      }
    }
    \ElseIf{$2^{d-2} \leq w$}
    {
       Construct a binary tree realizing $s_0 \land \dotsc \land s_{n-1}$. \; \label{algo large w start}
       Recursively realize $f((), t)$. \; \label{algo large w realization}
       Obtain a circuit for $f(s, t)$ via $f(s, t) = (s_0 \land \dotsc \land s_{n-1}) \land f((), t)$. \; \label{algo large w end}
    }
    \Else
    {
       \If{$W(t_0) > \Ceta \frac{2^{d-2} - w}{(d-1)\log_2(d-1)}$}
       {\label{algo small w start}
          Construct a binary symmetric tree realizing $s_0 \land \dotsc \land s_{n-1} \land t_0$. \;
          Recursively realize $f^*((), (t_1, \dotsc, t_{m-1}))$. \; \label{algo case 1 rec real}
          Obtain a circuit for $f(s, t)$ via $f(s, t) = (s_0 \land \dotsc \land s_{n-1} \land t_0) \land f^*((), (t_1, t_2, \dotsc, t_{m-1}))$. \; \label{algo case 1}
       }
       \Else
       {
          Choose a maximum odd-length prefix $t'$ of $t$ with $W(t') \leq \Ceta \frac{2^{d-2} - w}{(d-1)\log_2(d-1)}$.\; \label{compute t'}
          Set $t'' := t \backslash t'$.\;
          Set $\tilde t := t' \cup t''_0 \cup t''_1$.\;
          Set $t^* := \begin{cases}
                            t'        & \text{if } t'' \text{ contains $\leq 2$ elements or } W(\tilde t) > \Ceta \frac{2^{d-2} - w}{(d-1)\log_2 (d-1)} + \frac{d-2}{d-1} \Lambda_{\tilde t}\,, \\
                            \tilde{t} & \text{otherwise}\,.
                           \end{cases}$\;
          Set $t^{**} := t \backslash t^*$.\;
          Recursively realize $f(s, t^*)$.\; \label{algo realization f(s, t^*)}
          \If{$t''$ contains at most $2$ elements}
          {
             Realize $f^*(\widehat{t^*}, t^{**})$ as a binary symmetric tree.\; \label{algo small w <= 2}
          }
          \ElseIf{$t^* = t'$ and $t''$ contains exactly $3$ elements}
          {
             Directly construct a circuit realizing $f^*(\widehat{t'}, t'') = (\widehat{t'} \lor t''_0) \lor (t''_{1} \land t''_{2})$.\; \label{algo small w = 3}
          }
          \Else
          {
             Recursively realize $f^*(\widehat{t^*}, t^{**})$.\; \label{algo small w recursion}
          }
          Obtain a circuit for $f(s, t)$ via $f(s, t) = f(s, t^*) \land f^*(\widehat{t^*}, t^{**})$.\; \label{algo small w end}
       }
    }
   \caption{Circuit Construction}
   \label{alg}
 \end{algorithm}

%% file: general_functions.tex
\section{More General Boolean Functions} \label{generalized aops}

In this section, we extend Theorem \ref{theorem::improved_delay2}
from \aop s to similar functions that do not alternate between $\land$~and $\lor$~regularly, but arbitrarily.

\begin{definition}
 We call a Boolean function of the form
 \[h(t, \circ_1, \dotsc, \circ_{m-1}) := t_0 \circ_1 \bigg(t_1 \circ_2 \Big(t_2 \circ_3 \big( \dotsc \circ_{m-2} (t_{m-2} \circ_{m-1} t_{m-1})\dotsc\big)\Big)\bigg)\,,\]
 where $t = (t_0, \dotsc, t_{m-1})$ are Boolean input variables
 and for each $i \in \{1, \dotsc, m-1\}$, the symbol $\circ_i$ denotes a two-input gate over the basis $\{\land, \lor\}$,
 a \emph{generalized \aop}.
\end{definition}

\input{generalized_aop_figure}
 \begin{theorem} \label{generalized proof}
  Given Boolean input variables $t = (t_0, \dotsc, t_{m-1})$ and gates $\circ_1, \dotsc, \circ_{m-1}$,
  there is a circuit realizing the generalized \aop~$h(t, \circ_1, \dotsc, \circ_{m-1})$
  with delay at most
  \[\log_2 (W) + \log_2 \log_2 (c+1) + \log_2 \log_2 \log_2 (c+1) + 5.3\,,\]
  size at most $10 (c+1) \log_2 (c+1) \log_2 \log_2 (c+1) + m - c - 1$ and
  maximum fanout at most
  $\log_2 (c+1) + \log_2 \log_2 (c+1) + \log_2 \log_2 \log_2 (c+1) + 3.3$,
  where $c$ denotes the number of changes between $\land$ and $\lor$ or vice versa.
 \proof
 We first prove the delay bound.
 We partition the inputs $t_0, \dotsc, t_{m-1}$ of $h$ into $c+1$ maximal groups $P_0, \dotsc, P_{c}$
 of consecutive inputs that feed the same kind of gate, see Figure~\ref{generalized aop figure}~(a).
 We denote the common gate type of the gates fed by the inputs $P_b$ by $G_b$.
 Note that for each $b \in \{0, \dotsc, c-1\}$, the group $P_b$ contains at least $1$ input,
 while $P_c$ contains at least $2$ inputs.

 For each $b \in \{0, \dotsc, c\}$,
 we can build a symmetric binary $G_b$-tree on the inputs $P_b$ using Huffman coding.
 By Remark~\ref{sym tree bound}, this yields a Boolean circuit $C_b$ with output $t'_b$ with delay
 \begin{equation} \label{P_b delay}
   a(t'_b) = \left\lceil \log_2\left(\sum_{t_i \in P_b} W(t_i)\right) \right\rceil\,.
 \end{equation}

 Denote the outputs of the circuits $C_0, \dotsc, C_{c}$ by $t'_0, \dotsc, t'_{c}$, respectively.
 Without loss of generality, we may assume that $\circ_1 = \land$.
 Then, we can express $h(t, \circ_1, \dotsc, \circ_{m-1})$ as an \aop~as follows:
 \begin{equation} \label{h aop}
   h(t, \circ_1, \dotsc, \circ_{m-1}) = g(t'_0, \dotsc, t'_{c})
 \end{equation}
 In Figure~\ref{generalized aop figure}~(b), you can see the circuit arising from
 the generalized~\AOP~in Figure~\ref{generalized aop figure}~(a) in this way:
 The yellow gates are used for the circuits $C_b$ on the input groups $P_b$;
 and their outputs feed an~\aop~drawn with red~\AND~and green~\OR~gates.

 Write $W' := \sum_{b = 0}^{c} W(t'_i)$.
 Theorem \ref{theorem::improved_delay2} yields a circuit $C$ for $g((t'_0, \dotsc, t'_{c}))$
 and thus also for $h(t, \circ_1, \dotsc, \circ_{m-1})$ with delay at most
 \begin{equation} \label{delay bound W'}
   \log_2 (W') + \log_2 \log_2 (c+1) + \log_2 \log_2 \log_2 (c+1) + 4.3\,.
 \end{equation}
 Note that the weight $W'$ can be bounded by
 { 
  \newcommand{\ca}[1]{&\customalign{==}{#1}}
  \begin{align*}
   W' = \sum_{b = 0}^{c} 2^{a(t'_b)}
      \stackrel{(\ref{P_b delay})}{\leq} \sum_{b = 0}^{c} 2^{\log_2\left(\sum_{t_i \in P_b} W(t_i)\right) + 1}
      = 2 \sum_{b = 0}^{c} \sum_{t_i \in P_b} W(t_i)
      = 2 W(t)\,.
  \end{align*}
  Hence, the delay of $C$ stated in (\ref{delay bound W'}) can be bounded by
 \begin{align*}
   \ca{} \log_2 (W') + \log_2 \log_2 (c+1) + \log_2 \log_2 \log_2 (c+1) + 4.3 \\
   \ca{\leq} \log_2 (W) + \log_2 \log_2 (c+1) + \log_2 \log_2 \log_2 (c+1) + 5.3\,,
 \end{align*}
 which finishes the proof of the delay bound.

 For bounding the size of the arising circuit,
 note that for each $b \in \{0, \dotsc, c\}$, the circuit $C_b$ has $|P_b| - 1$ gates.
 Together with the size bound for the circuit realizing the \aop~(\ref{h aop}) on $c+1$ inputs
 shown in Theorem~\ref{thm:size_fanout},
 we obtain a total size of at most
 \begin{align*}
  \ca{} \sum_{b = 0}^c (|P_b| - 1) + 10 (c+1) \log_2 (c+1) \log_2 \log_2 (c+1) \\
  \ca{=} m - c - 1 + 10 (c+1) \log_2 (c+1) \log_2 \log_2 (c+1)\,.
 \end{align*}

 Since in the circuits $P_b$, $b = 0, \dotsc, c+1$, every node has exactly one predecessor
 and each input of the \aop~(\ref{h aop}) occurs only once,
 the maximum fanout occurs in the circuit for the \aop~(\ref{h aop}).
 Due to Theorem~\ref{thm:size_fanout}, this fanout is thus at most
 \[\log_2 (c+1) + \log_2 \log_2 (c+1) + \log_2 \log_2 \log_2 (c+1) + 3.3\,. \qedhere\]

 } 
 \end{theorem}

%% file: generalized_aop_figure.tex
\begin{figure}
\centering
\providecommand{\subfigw}{0.44\textwidth}
\begin{minipage}{\subfigw}
\vspace{-.4cm}
\resizebox{\textwidth}{!}{%
\begin{tikzpicture}
\node[outer sep=0pt, scale = 1.6] (x0)  at (-1.5, 6.2){$t_0$};
\node[outer sep=0pt, scale = 1.6] (x1)  at (-0.5, 6.2){$t_1$};
\node[outer sep=0pt, scale = 1.6] (x2)  at (0.5,  6.2){$t_2$};
\node[outer sep=0pt, scale = 1.6] (x3)  at (1.5,  6.2){$t_3$};
\node[outer sep=0pt, scale = 1.6] (x4)  at (2.5,  6.2){$t_4$};
\node[outer sep=0pt, scale = 1.6] (x5)  at (3.5,  6.2){$t_5$};
\node[outer sep=0pt, scale = 1.6] (x6)  at (4.5,  6.2){$t_6$};
\node[outer sep=0pt, scale = 1.6] (x7)  at (5.5,  6.2){$t_7$};
\node[outer sep=0pt, scale = 1.6] (x8)  at (6.5,  6.2){$t_8$};
\node[outer sep=0pt, scale = 1.6] (x9)  at (7.5,  6.2){$t_9$};
\node[outer sep=0pt, scale = 1.6] (x10) at (8.5,  6.2){$t_{10}$};
\node[outer sep=0pt, scale = 1.6] (x11) at (9.5,  6.2){$t_{11}$};

\draw[thick, decorate,decoration={brace, amplitude=5pt, raise=10pt}] (x0.west) -- (x0.east)  node [midway, above, sloped, yshift=15pt, scale=1.6] {$P_0$};
\draw[thick, decorate,decoration={brace, amplitude=5pt, raise=10pt}] (x1.west) -- (x4.east)  node [midway, above, sloped, yshift=15pt, scale=1.6] {$P_1$};
\draw[thick, decorate,decoration={brace, amplitude=5pt, raise=10pt}] (x5.west) -- (x6.east)  node [midway, above, sloped, yshift=15pt, scale=1.6] {$P_2$};
\draw[thick, decorate,decoration={brace, amplitude=5pt, raise=10pt}] (x7.west) -- (x7.east)  node [midway, above, sloped, yshift=15pt, scale=1.6] {$P_3$};
\draw[thick, decorate,decoration={brace, amplitude=5pt, raise=10pt}] (x8.west) -- (x11.east) node [midway, above, sloped, yshift=15pt, scale=1.6] {$P_4$};

\node[fill=red, outer sep=0pt, and gate US, draw, logic gate inputs=nn, rotate=270, thick] at (9,5) (or8){};
\node[fill=red, outer sep=0pt, and gate US, draw, logic gate inputs=nn, rotate=270, thick] at (8,4) (and7){};
\node[fill=red, outer sep=0pt, and gate US, draw, logic gate inputs=nn, rotate=270, thick] at (7,3) (or1){};

\node[fill=green, outer sep=0pt, or gate US, draw, logic gate inputs=nn, rotate=270, thick] at (6,2) (and1){};

\node[fill=red, outer sep=0pt, and gate US, draw, logic gate inputs=nn, rotate=270, thick] at (5,1) (or4){};
\node[fill=red, outer sep=0pt, and gate US, draw, logic gate inputs=nn, rotate=270, thick] at (4,0) (and2){};

\node[fill=green, outer sep=0pt, or gate US, draw, logic gate inputs=nn, rotate=270, thick] at (3,-1) (or5){};
\node[fill=green, outer sep=0pt, or gate US, draw, logic gate inputs=nn, rotate=270, thick] at (2,-2) (and6){};
\node[fill=green, outer sep=0pt, or gate US, draw, logic gate inputs=nn, rotate=270, thick] at (1,-3) (or9){};
\node[fill=green, outer sep=0pt, or gate US, draw, logic gate inputs=nn, rotate=270, thick] at (0,-4) (and10){};

\node[fill=red, outer sep=0pt, and gate US, draw, logic gate inputs=nn, rotate=270, thick] at (-1,-5) (or11){};

\draw[thick] (or1.output) -- (and1.input 1);
\draw[thick] (x8) -- (or1.input 2);
\draw[thick] (x7) -- (and1.input 2);
\draw[thick] (x9) -- (and7.input 2);

\draw[thick] (or4.output) -- (and2.input 1);
\draw[thick] (or8.output) -- (and7.input 1);

\draw[thick] (x5) -- (and2.input 2);
\draw[thick] (x10) -- (or8.input 2);
\draw[thick] (x11) -- (or8.input 1);

\draw[thick] (x6) -- (or4.input 2);
\draw[thick] (and1.output) -- (or4.input 1);

\draw[thick] (and2.output) -- (or5.input 1);
\draw[thick] (or5.output) -- (and6.input 1);

\draw[thick] (x4) -- (or5.input 2);
\draw[thick] (x3) -- (and6.input 2);

\draw[thick] (x2) -- (or9.input 2);
\draw[thick] (and6.output) -- (or9.input 1);
\draw[thick] (and7.output) -- (or1.input 1);

\draw[thick] (or9.output) -- (and10.input 1);
\draw[thick] (x1) -- (and10.input 2);
\draw[thick] (and10.output) -- (or11.input 1);
\draw[thick] (x0) -- (or11.input 2);

\end{tikzpicture}}
\caption*{(a) \, A generalized \aop.}
\end{minipage}
\hfill
\begin{minipage}{\subfigw}
\resizebox{\textwidth}{!}{%
\begin{tikzpicture}
\node[outer sep=0pt, scale = 1.6] (x0)  at (-1.5, 6.2){$t_0$};
\node[outer sep=0pt, scale = 1.6] (x1)  at (-0.5, 6.2){$t_1$};
\node[outer sep=0pt, scale = 1.6] (x2)  at (0.5,  6.2){$t_2$};
\node[outer sep=0pt, scale = 1.6] (x3)  at (1.5,  6.2){$t_3$};
\node[outer sep=0pt, scale = 1.6] (x4)  at (2.5,  6.2){$t_4$};
\node[outer sep=0pt, scale = 1.6] (x5)  at (3.5,  6.2){$t_5$};
\node[outer sep=0pt, scale = 1.6] (x6)  at (4.5,  6.2){$t_6$};
\node[outer sep=0pt, scale = 1.6] (x7)  at (5.5,  6.2){$t_7$};
\node[outer sep=0pt, scale = 1.6] (x8)  at (6.5,  6.2){$t_8$};
\node[outer sep=0pt, scale = 1.6] (x9)  at (7.5,  6.2){$t_9$};
\node[outer sep=0pt, scale = 1.6] (x10) at (8.5,  6.2){$t_{10}$};
\node[outer sep=0pt, scale = 1.6] (x11) at (9.5,  6.2){$t_{11}$};

\draw[thick, decorate,decoration={brace, amplitude=5pt, raise=10pt}] (x0.west) -- (x0.east)  node [midway, above, sloped, yshift=15pt, scale = 1.6] {$P_0$};
\draw[thick, decorate,decoration={brace, amplitude=5pt, raise=10pt}] (x1.west) -- (x4.east)  node [midway, above, sloped, yshift=15pt, scale = 1.6] {$P_1$};
\draw[thick, decorate,decoration={brace, amplitude=5pt, raise=10pt}] (x5.west) -- (x6.east)  node [midway, above, sloped, yshift=15pt, scale = 1.6] {$P_2$};
\draw[thick, decorate,decoration={brace, amplitude=5pt, raise=10pt}] (x7.west) -- (x7.east)  node [midway, above, sloped, yshift=15pt, scale = 1.6] {$P_3$};
\draw[thick, decorate,decoration={brace, amplitude=5pt, raise=10pt}] (x8.west) -- (x11.east) node [midway, above, sloped, yshift=15pt, scale = 1.6] {$P_4$};

\node[fill=yellow, outer sep=0pt, or gate US, draw, logic gate inputs=nn, rotate=270, thick] at ( 0.5,5) (or1){};
\node[fill=yellow, outer sep=0pt, or gate US, draw, logic gate inputs=nn, rotate=270, thick] at ( 1.5,5) (or2){};
\node[fill=yellow, outer sep=0pt, or gate US, draw, logic gate inputs=nn, rotate=270, thick] at ( 1,3.8) (or3){};

\node[fill=yellow, outer sep=0pt, and gate US, draw, logic gate inputs=nn, rotate=270, thick] at (4,5) (and1){};

\node[fill=yellow, outer sep=0pt, and gate US, draw, logic gate inputs=nn, rotate=270, thick] at (8,5) (and2){};
\node[fill=yellow, outer sep=0pt, and gate US, draw, logic gate inputs=nn, rotate=270, thick] at (8.5,3.8) (and3){};
\node[fill=yellow, outer sep=0pt, and gate US, draw, logic gate inputs=nn, rotate=270, thick] at (7.5,2.8) (and4){};

\node[fill=green, outer sep=0pt, or gate US, draw, logic gate inputs=nn, rotate=270, thick] at ( 6.5,1.8) (or4){};
\node[fill=red, outer sep=0pt, and gate US, draw, logic gate inputs=nn, rotate=270, thick] at (5.5,0.8) (and5){};
\node[fill=green, outer sep=0pt, or gate US, draw, logic gate inputs=nn, rotate=270, thick] at ( 4.5,-.2) (or5){};
\node[fill=red, outer sep=0pt, and gate US, draw, logic gate inputs=nn, rotate=270, thick] at (3.5,-1.2) (and6){};

\phantom{\node[fill=red, outer sep=0pt, and gate US, draw, logic gate inputs=nn, rotate=270, thick] at (-1,-5) (or11){};}

\draw[thick] (x1) -- (or1.input 2);
\draw[thick] (x3) -- (or1.input 1);

\draw[thick] (x2) -- (or2.input 2);
\draw[thick] (x4) -- (or2.input 1);

\draw[thick] (or1.output) -- (or3.input 2);
\draw[thick] (or2.output) -- (or3.input 1);

\draw[thick] (x5) -- (and1.input 2);
\draw[thick] (x6) -- (and1.input 1);

\draw[thick] (x9) -- (and2.input 2);
\draw[thick] (x10) -- (and2.input 1);

\draw[thick] (and2.output) -- (and3.input 2);
\draw[thick] (x11) -- (and3.input 1);

\draw[thick] (x8) -- (and4.input 2);
\draw[thick] (and3.output) -- (and4.input 1);

\draw[thick] (x7) -- (or4.input 2);
\draw[thick] (and4.output) -- (or4.input 1);

\draw[thick] (and1.output) -- (and5.input 2);
\draw[thick] (or4.output) -- (and5.input 1);

\draw[thick] (or3.output) -- (or5.input 2);
\draw[thick] (and5.output) -- (or5.input 1);

\draw[thick] (x0) -- (and6.input 2);
\draw[thick] (or5.output) -- (and6.input 1);

\end{tikzpicture}}
\caption*{(b) \, Equivalent circuit after performing Huffman coding on the groups $P_b$.}
\end{minipage}
\caption{Illustration of the proof of Theorem \ref{generalized proof}.}
\label{generalized aop figure}
\end{figure}